\documentclass[10pt,letterpaper]{amsart}
\usepackage{amsfonts,amsmath,color,ifpdf,latexsym,graphicx,subfig,url}
\title[Characterization of Curved Creases and Rulings]
      {Characterization of Curved Creases and Rulings: \\
       Design and Analysis of Lens Tessellations}

\author[E. Demaine]{Erik D. Demaine}
\address[E. Demaine and M. Demaine]{Computer Science and Artificial
  Intelligence Laboratory, Massachusetts Institute of Technology,
  32 Vassar St., Cambridge, MA 02139, USA}
\email{demaine@mit.edu}
\thanks{E. Demaine and M. Demaine supported in part by NSF ODISSEI grant EFRI-1240383 and NSF Expedition grant CCF-1138967.}
\author[M. Demaine]{Martin L. Demaine}
\author[D. Huffman]{David A. Huffman}
\address[D. Huffman]{Department of Computer Science, University of California,
  Santa Cruz, CA 95064, USA.}
\author[D. Koschitz]{Duks Koschitz}
\address[D. Koschitz]{School of Architecture, Pratt Institute, 200 Willoughby Ave.,
  Brooklyn, NY 11205, USA.}
\email{duks@pratt.edu}
\thanks{D. Koschitz performed this research while at MIT}
\author[T. Tachi]{Tomohiro Tachi}
\address[T. Tachi]{Department of General Systems Studies, The University of Tokyo,
  3-8-1 Komaba, Meguro-Ku, Tokyo 153-8902, Japan}
\email{tachi@idea.c.u-tokyo.ac.jp}
\thanks{T. Tachi supported by the JST Presto program.}

\date{}

\newif\ifabstract
\abstracttrue
\newif\iffull
\ifabstract \fullfalse \else \fulltrue \fi

\usepackage
  [breaklinks,bookmarks,bookmarksnumbered,bookmarksopen,bookmarksopenlevel=2]
  {hyperref}
{\makeatletter \hypersetup{pdftitle={\@title}}}

\usepackage{graphicx}
\graphicspath{
 {./images/}
}

{\makeatletter
 \gdef\xxxmark{%
   \expandafter\ifx\csname @mpargs\endcsname\relax 
     \expandafter\ifx\csname @captype\endcsname\relax 
       \marginpar{xxx}
     \else
       xxx 
     \fi
   \else
     xxx 
   \fi}
 \gdef\xxx{\@ifnextchar[\xxx@lab\xxx@nolab}
 \long\gdef\xxx@lab[#1]#2{\textbf{[\xxxmark #2 ---{\sc #1}]}}
 \long\gdef\xxx@nolab#1{\textbf{[\xxxmark #1]}}
 \long\gdef\xxx@lab[#1]#2{}\long\gdef\xxx@nolab#1{}%
}

{\makeatletter \gdef\fps@figure{!htbp}}

\let\realbfseries=\bfseries
\def\bfseries{\realbfseries\boldmath}

\newtheorem{theorem}{Theorem}
\newtheorem{lemma}[theorem]{Lemma}
\newtheorem{corollary}[theorem]{Corollary}

\def\term{\emph}
\let\epsilon=\varepsilon
\def\vec{\mathbf}
\def\vechat#1{\mathbf{\hat{#1}}}

\def\R{{\mathbb R}}

\begin{document}
\maketitle

\begin{abstract}

We describe a general family of curved-crease folding tessellations consisting
of a repeating ``lens'' motif formed by two convex curved arcs. The third
author invented the first such design in 1992, when he made both a sketch of
the crease pattern and a vinyl model (pictured below). Curve fitting suggests
that this initial design used circular arcs. We show that in fact the curve
can be chosen to be any smooth convex curve without inflection point. We
identify the ruling configuration through qualitative properties that a curved
folding satisfies, and prove that the folded form exists with no additional
creases, through the use of differential geometry.

\end{abstract}



\section{Introduction}

The past two decades have seen incredible advances in applying
mathematics and computation to the analysis and design of origami
made by straight creases.  But we lack many similar theorems and
algorithms for origami made by curved creases.

In this paper, we develop several basic tools (definitions and theorems)
for curved-crease origami.  These tools in particular characterize the
relationship between the crease pattern and rule lines/segments, and relate
creases connected by rule segments.  Some of these tools have been developed
before in other contexts
(e.g., \cite{Fuchs-Tabachnikov-1999,Fuchs-Tabachnikov-2007-developable,Huffman-1976}),
but have previously lacked a careful analysis of the
levels of smoothness ($C^1$, $C^2$, etc.)\ and other assumptions required.
Specific high-level properties we prove include:
\begin{enumerate}
\item Regions between creases decompose into noncrossing rule segments,
      which connect from curved crease to curved crease,
      and planar patches (a result from \cite{Hypar}).
\item The osculating plane of a crease bisects the two adjacent
      surface tangent planes (when they are unique).
\item A curved crease with an incident cone ruling
      (a continuum of rule segments at a point) cannot fold smoothly:
      it must be kinked at the cone ruling.
\item Rule segments on the convex side of a crease bend mountain/valley
      the same as the crease, and rule segments on the concave side
      of a crease bend mountain/valley opposite from the crease.
\item If two creases are joined by a rule segment on their concave sides,
      or on their convex sides,
      then their mountain/valley assignments must be equal.
      If the rule segment is on the convex side of one crease and the
      concave side of the other crease, then the mountain/valley assignments
      must be opposite.
\end{enumerate}

We apply these tools to analyze one family of designs called the
\emph{lens tessellation}.  Figure~\ref{huffman} shows an example
originally designed and folded by the third author in 1992, and now
modeled digitally.  We prove that this curved crease pattern folds into 3D,
with the indicated rule segments, when the ``lens'' is \emph{any}
smooth convex curve.
We also show that the model is ``rigidly foldable'', meaning that it can be
continuously folded without changing the ruling pattern.

\begin{figure}
  \centering
  \subfloat[Huffman's original hand-drawn sketch of crease pattern of
            lens design (1992).]
    {\includegraphics[width=0.44\textwidth]{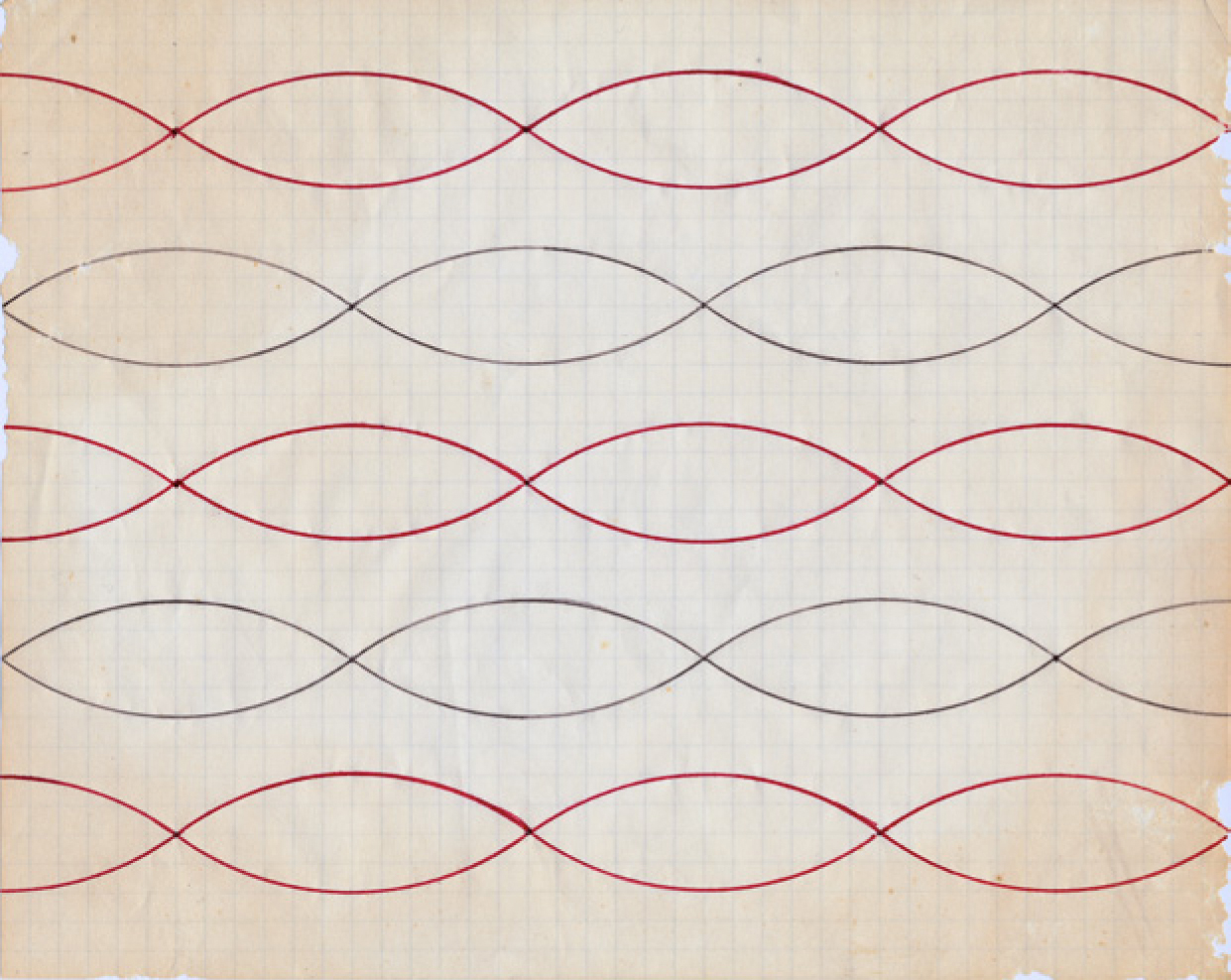}}
  \hfill
  \subfloat[Computer-drawn crease pattern of lens design.]
    {\includegraphics[width=0.53\textwidth]{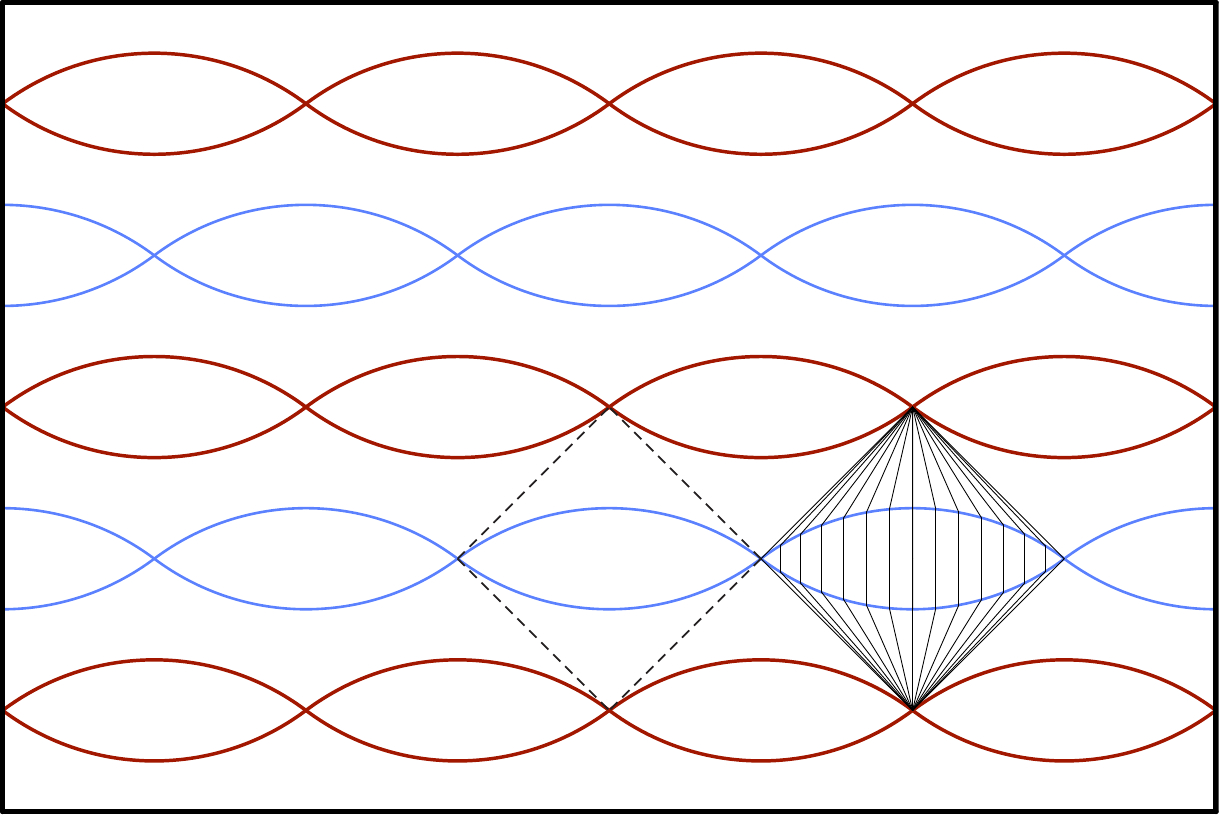}}

  \subfloat[Huffman's original hand-folded vinyl model (1992).
            Photo by Tony Grant.]
    {\includegraphics[width=0.44\textwidth]{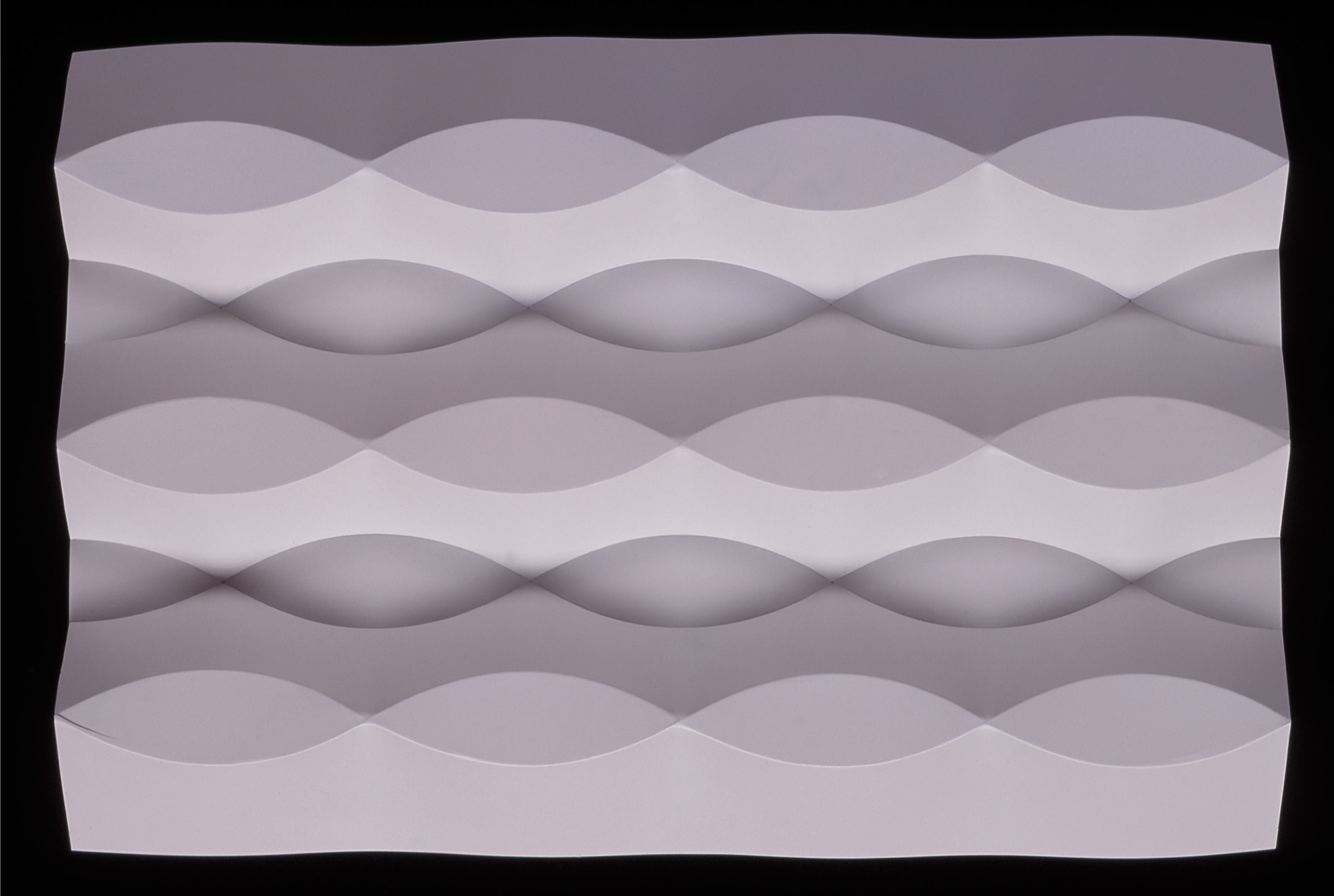}}
  \hfill
  \subfloat[Computer-simulated 3D model using Tachi's Freeform Origami software.]
    {\includegraphics[width=0.53\textwidth]{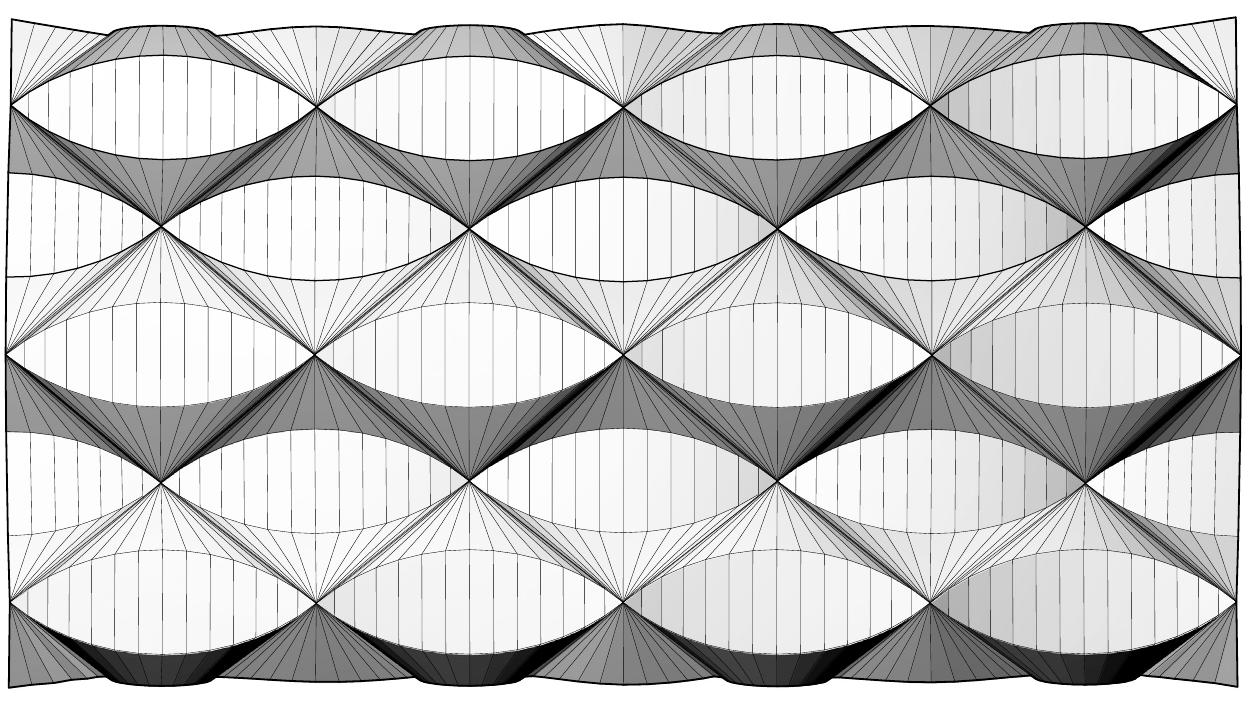}}
  \caption{\label{huffman}
    Lens tessellation: 1992 original (left) and digital reconstruction (right).}
\end{figure}

The 3D configuration of the curved folding is solved through identifying the
correspondence between pairs of points connected by rule segments, using
the qualitative properties described above.
These properties separate the tessellation into independent kite-shaped tiles
and force the rulings between the lenses to be particular cones with their
apices coinciding with the vertices of the tiling. The ruling inside each lens
is free (can twist), but assuming no twist or global planarity/symmetry, is
cylindrical (vertical rule segments). The tiling exists by rotation/reflection
of the 3D model of each kite around its four straight boundary edges. From the
tiling symmetry, each tile edge has a common tangent to its neighbors
regardless of the type of curves, as long as it is a convex curve.

The rest of this paper is organized as follows.
Section~\ref{Curves} introduces some basic notation for 2D and 3D curves.
Section~\ref{Foldings} defines creases, crease patterns, foldings,
rule segments, cone ruling, orientation of the paper, and surface normals
(and analyzes when they exist).
Section~\ref{Bisection Property} proves that the powerful bisection
property---the osculating plane of a crease bisects the two adjacent surface
tangent planes---and uses it to rule out some strange situations such as
rule segments tangent to creases or zero-length rule segments.
Section~\ref{Smooth Folding} characterizes smooth folding: a crease is
folded $C^1$ if and only if it is folded $C^2$ if and only if
there are no incident cone rulings.
Section~\ref{Mountains and Valleys} defines mountains and valleys for
both creases and the bending of rule segments, and relates the two.
Finally, Section~\ref{Lens Tessellation} uses all these tools to analyze
lens tessellations, proving a necessary and sufficient condition on their
foldability.


\section{Curves}
\label{Curves}

In this section, we define some standard parameterizations of curves
in 2D and 3D, which we will use in particular for describing creases
in the unfolded paper and folded state.  Our notation introduces a helpful
symmetry between 2D (unfolding) and 3D (folding): lower case indicates 2D,
while upper case indicates the corresponding notion in 3D.

\subsection{2D Curves}

Consider an arclength-parameterized $C^2$ 2D curve
$\vec x : (0,\ell) \to \mathbb R^2$
(or in any metric 2-manifold).
For $s \in (0,\ell)$, define the (unit) \emph{tangent} at $s$ by
$$ \vec t(s) = {d \vec x(s) \over d s}. $$
Define the \emph{curvature}
$$ k(s) = \left\| {d \vec t(s) \over d s} \right\|. $$
In particular, call the curve \emph{curved} at $s$
if its curvature $k(s)$ is nonzero.
In this case, define the (unit) \emph{normal} at $s$ by
$$ \vec n(s) = \left. {d \vec t(s) \over d s} \right/ k(s). $$
The curve is \emph{curved} (without qualification) if it is curved
at all $s \in (0,\ell)$.

Define the \emph{convex side} at $s$ to consist of directions
having negative dot product with $\vec n(s)$;
and define the \emph{concave side} at $s$ to consist of directions
having positive dot product with $\vec n(s)$.

\subsection{3D Curves}

For an arclength-parameterized $C^2$ space curve $\vec X : [0,\ell] \to \R^3$,
and for a parameter $s \in [0,\ell]$ inducing a point $\vec X(s)$,
define the (unit) \emph{tangent}
$$ \vec T(s) = {d \vec X(s) \over d s}. $$
Define the \emph{curvature}
$$ K(s) = \left\| {d \vec T(s) \over d s} \right\|. $$
In particular, call the curve \emph{curved} at $s$
if its curvature $K(s)$ is nonzero
(and \emph{curved} without qualification if it is curved
at all $s \in (0,\ell)$).
In this case, define the (unit) \emph{normal} at $s$ by
$$ \vec N(s) = \left. {d \vec T(s) \over d s} \right/ K(s); $$
define the (unit) \emph{binormal}
$$ \vec B(s) = \vec T(s) \times \vec N(s); $$
and define the \emph{torsion}
$$ \tau(s) = -{d \vec B(s) \over d s} \cdot \vec N(s). $$

Equivalently, these definitions follow from the Frenet--Serret formulas:
$$ 
\begin{bmatrix}
     0     & K(s)     & 0       \\
     -K(s) & 0        & \tau(s) \\
     0     & -\tau(s) & 0       \\
   \end{bmatrix}
 \cdot
   \begin{bmatrix}
     \vec T(s) \\
     \vec N(s) \\
     \vec B(s) \\
   \end{bmatrix}
 =
   {d \over d s}
   \begin{bmatrix}
     \vec T(s) \\
     \vec N(s) \\
     \vec B(s) \\
   \end{bmatrix}.
$$

\begin{lemma}
  For any curved $C^2$ 3D curve $\vec X(s)$,
  the Frenet frame $(\vec T(s), \vec N(s), \vec B(s))$ and
  curvature $K(s)$ exist and are continuous.
\end{lemma}

\begin{proof}
  Because $\vec X(s)$ is differentiable, $\vec T(s)$ exists.
  Because $\vec X(s)$ is twice differentiable, $K(s)$ exists,
  and because $\vec X(s)$ is $C^2$, $K(s)$ is continuous.
  Because the curve is curved, $K(s) \neq 0$,
  so we do not divide by $0$ in computing $\vec N(s)$,
  and thus $\vec N(s)$ exists and is continuous.
  The cross product in $\vec B(s)$ exists and is continuous
  because $\vec T(s)$ and $\vec N(s)$ are guaranteed to be normalized
  (hence nonzero) and orthogonal to each other (hence not parallel).
\end{proof}

The same lemma specializes to 2D, by dropping the $\vec B(s)$ part:

\begin{corollary}
  For any curved $C^2$ 2D curve $\vec x(s)$,
  the frame $(\vec t(s), \vec n(s))$
  curvature $k(s)$ exist and are continuous.
\end{corollary}


\section{Foldings}
\label{Foldings}

The following definitions draw from \cite{Hypar,GFALOP}.

We start with 2D (unfolded) notions.
A \term{piece of paper} is an open 2-manifold embedded in $\R^2$.
A \term{crease} $\vec x$ is a $C^2$ 2D curve contained in the piece of paper
and not self-intersecting (i.e., not visiting the same point twice).
A \term{crease point} is a point $\vec x(s)$ on the relative interior
of the crease (excluding endpoints).
The endpoints of a crease are \term{vertices}.
A \term{crease pattern} is a collection of creases that meet only
at common vertices.
Equivalently, a crease pattern is an embedded planar graph, where each
edge is embedded as a crease.
This definition effectively allows piecewise-$C^2$ curves,
by subdividing the edge in the graph with additional vertices;
``creases'' are the resulting $C^2$ pieces.
A \term{face} is a maximal open region of the piece of paper
not intersecting any creases or vertices.

Now we proceed to 3D (folded) notions.
A \term{(proper) folding} of a crease pattern is a
piecewise-$C^2$ isometric embedding of the piece of paper into 3D
that is $C^1$ on every face and not $C^1$ at every crease point and vertex.
Here \term{isometric} means that intrinsic path lengths are preserved
by the mapping; and \term{piecewise-$C^2$} means that the folded image
can be decomposed into a finite complex of $C^2$ open regions
joined by points and $C^2$ curves.
We use the terms \term{folded crease}, \term{folded vertex},
\term{folded face}, and \term{folded piece of paper} to refer to the image of
a crease, vertex, face, and entire piece of paper under the folding map.
Thus, each folded face subdivides into a finite complex of $C^2$ open regions
joined by points called \term{folded semivertices} and
$C^2$ curves called \term{folded semicreases}.
Each folded crease $\vec X(s)$ can be subdivided into a finite sequence of
$C^2$ curves joined by $C^1$ points called \emph{semikinks} and
not-$C^1$ points called \emph{kinks}.
(Here $C^1$/not-$C^1$ is a property measured of the crease $\vec X(s)$;
 crease points are necessarily not $C^1$ on the folded piece of paper.)
In fact, semivertices do not exist \cite[Corollary~2]{Hypar},
and neither do semikinks (Corollary~\ref{no semikink} below).

\begin{lemma}
\label{no straight}
  A curved crease $\vec x(s)$ folds into a 3D curve $\vec X(s)$
  that contains no line segments
  (and thus is curved except at kinks and semikinks).
\end{lemma}

\begin{proof}
  Suppose $X(s)$ is a 3D line segment for $s \in [s_1, s_2]$.
  Then the distance between $\vec X(s_1)$ and $\vec X(s_2)$ as measured on the
  folded piece of paper is the length of this line segment,
  i.e., the arc length of $\vec X$ over $s \in [s_1, s_2]$ which,
  by isometry, equals the arc length of $\vec x$ over $s \in [s_1, s_2]$.
  However, in the 2D piece of paper, there is a shorter path connecting
  $\vec x(s_1)$ and $\vec x(s_2)$ because the 2D crease is curved
  (and not on the paper boundary, because the paper is an open set),
  contradicting isometry.
%
\end{proof}

\subsection{Developable Surfaces}

A folded face is also known as an uncreased developable surface:
it is \term{uncreased} in the sense that it is $C^1$,
and \term{developable} in the sense that every point $p$ has a
neighborhood isometric to a region in the plane.
The following theorem from \cite{Hypar} characterizes what
uncreased developable surfaces look like:

\begin{theorem}[Corollaries 1--3 of \cite{Hypar}] \label{rule lines}
  Every interior point $p$ of an uncreased developable surface $M$ not
  belonging to a planar neighborhood belongs to a unique rule segment $C_p$.
  The rule segment's endpoints are on the boundary of~$M$.
  In particular, every semicrease is such a rule segment.
\end{theorem}

\begin{corollary} \label{uncreased region}
  Any folded face decomposes into planar regions and
  nonintersecting rule segments (including semicreases)
  whose endpoints lie on creases.
\end{corollary}


For a folded piece of paper, we use the term \emph{(3D) rule segment} for
exactly these segments $C_p$ computed for each folded face,
for all points $p$ that are not folded vertices, not folded crease points,
and not belonging to a planar neighborhood.
In particular, we view the interior of planar regions as not containing any
rule segments (as they would be ambiguous); however,
the boundaries of planar regions are considered rule segments.
As a consequence, all rule segments have a neighborhood that is nonplanar.

For each 3D rule segment in the folded piece of paper,
we can define the corresponding 2D rule segment by the inverse mapping.
By isometry, 2D rule segments are indeed line segments.

Define a \emph{cone ruling} at a crease point $\vec x(s)$ to be a fan of
2D rule segments emanating from $\vec x(s)$ in a positive-length interval of
directions $[\theta_1, \theta_2]$.

\subsection{Orientation}

We orient the piece of paper in the $xy$ plane
by a consistent normal $\vec e_z$ (in the $+z$ direction)
called the \emph{top side}.
This orientation defines, for a 2D crease $\vec x = \vec x(s)$
in the crease pattern, a \emph{left normal}
$\vechat n(s) = \vec e_z \times \vec t(s)$.
Where $\vec x(s)$ is curved and thus $\vec n(s)$ is defined,
we have $\vechat n(s) = \pm \vec n(s)$ where the sign specifies whether
the left or right side corresponds to the convex side of the curve.
We can also characterize a 2D rule segment incident to $\vec x(s)$
as being \emph{left} of $\vec x$ when the vector emanating from $\vec x(s)$
has positive dot product with $\vechat n(s)$, and \emph{right} of $\vec x$
when it has negative dot product.
(In Lemma~\ref{no tangent} below, we prove that no rule segment is tangent
 to a crease, and thus every rule segment is either left or right of the crease.)

We can also define the \emph{signed curvature} $\hat k(s)$ to flip sign where
$\vechat n(s)$ does: $\hat k(s) \vechat n(s) = k(s) \vec n(s)$.
Then $\hat k(s)$ is positive where the curve turns left and
negative where the curve turns right (relative to the top side).


\subsection{Unique Ruling}

Call a crease point $\vec x(s)$ \emph{uniquely ruled on the left}
if there is exactly one rule segment left of $\vec x(s)$;
symmetrically define \emph{uniquely ruled on the right};
and define \emph{uniquely ruled} to mean uniquely ruled on
both left and right.

By Corollary~\ref{uncreased region},
there are two possible causes for a crease point $\vec x(s)$ to be not
uniquely ruled (say on the left).
First, there could be one or more cone rulings (on the left) at $\vec x(s)$.
Second, there could be one or more planar 3D regions incident to $\vec X(s)$
(which, in 2D, lie on the left of $\vec x(s)$, meaning the points have
positive dot product with $\vechat n(s)$).

\begin{figure}
  \centering
  \includegraphics[width=0.6\linewidth]{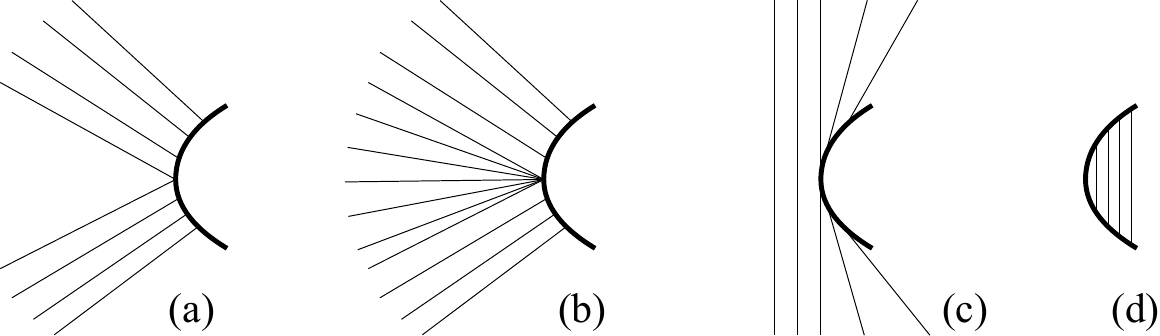}
  \caption{Possibilities for a crease to be not uniquely ruled.}
  \label{fig:not-uniquely-ruled-cases}
\end{figure}

One special case of unique ruling is when a rule segment is tangent to
a curved crease.  Ultimately, in Lemma~\ref{no tangent}, we will prove
that this cannot happen, but for now we need that the surface normals
remain well-defined in this case.
There are two subcases depending on whether the rule segment is on the
convex or concave side of the crease, as in
Figures~\ref{fig:not-uniquely-ruled-cases}(c) and~(d).
The rule segment's direction in 3D and surface normal vector remain
well-defined in this case, by taking limits of nearby rule segments.
In the concave subcase~(d), we take the limit of rule segments on the same
side of the curve.
In the convex subcase~(c), the rule segment splits the surface locally into two
halves, and we take the limit of rule segments in the half not containing
the crease.
Because the surface normals are thus well-defined, we do not need to
distinguish this case in our proofs below.

%

Call a crease point $\vec x(s)$ \emph{cone free} if there are no cone rulings
at $\vec x(s)$; similarly define \emph{cone free on the left/right}.
Such a point may still have a planar region, but only one:

\begin{lemma} \label{only one planar}
  If a crease point $\vec x(s)$ is cone free, then it has at most one
  planar region on each side.
\end{lemma}

\begin{proof}
  Refer to Figure~\ref{fig:one-planar-region}.
  Suppose $\vec x(s)$ had at least two planar regions on, say, the left side.
  Order the regions clockwise around $\vec x(s)$, and pick two adjacent
  planar regions $R_1$ and $R_2$.  By Corollary~\ref{uncreased region},
  the wedge with apex $\vec x(s)$ between $R_1$ and $R_2$ must be covered by
  rule segments.  But by Theorem~\ref{rule lines}, a rule segment cannot
  have its endpoints on the boundaries of $R_1$ and $R_2$, as it must
  extend all the way to creases.  Thus the only way to cover the wedge
  locally near $\vec x(s)$ is to have a cone ruling at $\vec x(s)$.
\end{proof}
\begin{figure}
  \centering
  \includegraphics[width=0.3\linewidth]{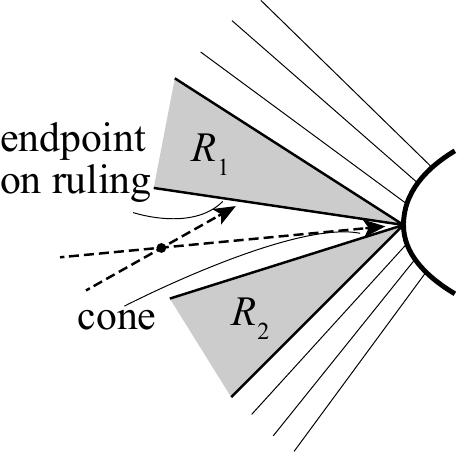}
  \caption{Two adjacent planar regions at a point.}
  \label{fig:one-planar-region}
\end{figure}

\subsection{Surface Normals}

In 3D, the orientation defines a top-side normal vector at every $C^1$ point.%
\footnote{For example, take infinitesimally small triangles around the point,
  oriented counterclockwise in 2D, and compute their normals in 3D.}
For a crease point $\vec X(s)$ that is cone free on the left,
we can define a unique \emph{left surface normal} $\vec P_L(s)$.
First, if there is a planar region on the left of $\vec X(s)$,
then by Lemma~\ref{only one planar} there is only one such planar region,
and we define $\vec P_L(s)$ to be
the unique top-side normal vector of the planar region.
Otherwise,
$\vec X(s)$ is uniquely ruled on the left, and we define $P_L(s)$ to be
the top-side surface normal vector which is constant along
this unique rule segment.
(As argued above, this definition makes sense even when the rule segment
 is a zero-length limit of rule segments.)
Similarly, we can define the right surface normal $\vec P_R(s)$
when $\vec X(s)$ is cone free on the right.



%
%
%

\section{Bisection Property}
\label{Bisection Property}

In this section, we prove that, at a cone-free folded curved crease,
the binormal vector bisects the left and right surface normal vectors,
which implies that the osculating plane of the crease bisects the
two surface tangent planes.  Proving this bisection property requires
several steps along the way, and has several useful consequences.

%

\subsection{$C^2$ Case}

First we prove the bisection property at $C^2$ crease points,
using the following simple lemma:

\begin{lemma} \label{geodesic curvature}
  For a $C^2$ folded curved crease $\vec X(s)$ that is cone-free on the left,
  $$(K(s) \vec N(s)) \cdot (\vec P_L(s) \times \vec T(s)) = \vechat k(s).$$
  For a $C^2$ folded curved crease $\vec X(s)$ that is cone-free on the right,
  $$(K(s) \vec N(s)) \cdot (\vec P_R(s) \times \vec T(s)) = \vechat k(s).$$
\end{lemma}

\begin{proof}
  We prove the left case; the right case is symmetric.
  The left-hand side is known as the geodesic curvature at $\vec X(s)$
  on surface $S_L$, and is known to be invariant under isometry\xxx{cite}.
  In the unfolded 2D state, the geodesic curvature is
  $$(k(s) \vec n(s)) \cdot (\vec e_z \times \vec t(s))
  = (k(s) \vec n(s)) \cdot \vechat n(s) = \hat k(s).$$
\end{proof}

\begin{lemma} \label{osculating plane bisects C^2}
  For a $C^2$ cone-free folded curved crease $\vec X(s)$,
  $\vec B(s)$ bisects $\vec P_L(s)$ and $\vec P_R(s)$.
  In particular, the tangent planes of the surfaces on both sides
  of $\vec X(s)$ form the same angle with the osculating plane.
\end{lemma}

\begin{proof}
A $C^2$ cone-free folded curved crease $\vec X(s)$ has unique left and right
surface normals $\vec P_L(s)$ and $\vec P_R(s)$.
By Lemma~\ref{geodesic curvature}, the left and right geodesic curvatures match:
$$(K(s) \vec N(s)) \cdot (\vec P_L(s) \times \vec T(s))
= (K(s) \vec N(s)) \cdot (\vec P_R(s) \times \vec T(s)).$$
The $K(s)$ scalars cancel, leaving a triple product:
$$\vec N(s) \cdot (\vec P_L(s) \times \vec T(s))
= \vec N(s) \cdot (\vec P_R(s) \times \vec T(s)),$$
which is equivalent to
$$\vec P_L(s) \cdot (\vec T(s) \times \vec N(s))
= \vec P_R(s) \cdot (\vec T(s) \times \vec N(s)).$$
Therefore $\vec B(s) = \vec T(s) \times \vec N(s)$ forms the same angle with
$\vec P_L(s)$ and $\vec P_R$.
Because $\vec B$, $\vec P_L$ and $\vec P_R$ lie in a common plane
orthogonal to $\vec T$, $\vec B$ bisects $\vec P_L$ and $\vec P_R$.
\end{proof}

\begin{figure}
  \centering
  \includegraphics[width=0.6\linewidth]{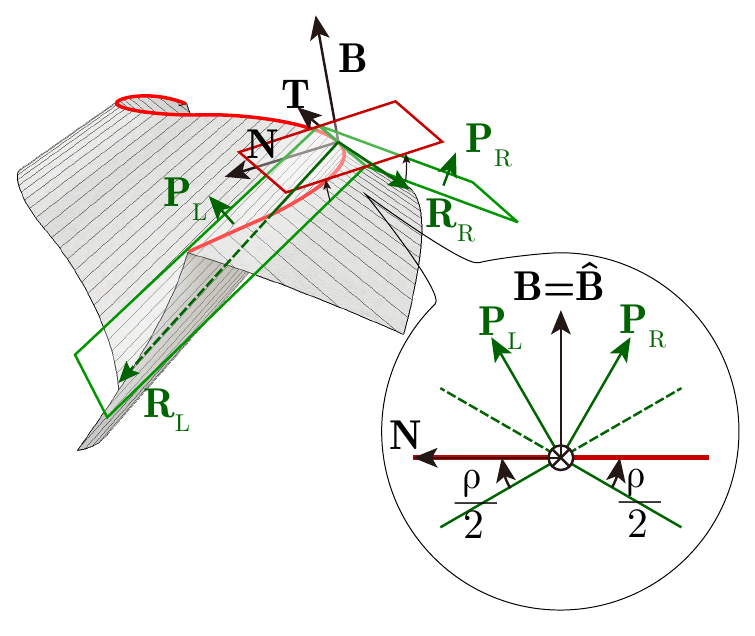}
  \caption{Binormal vector bisects surface normals.}
  \label{fig:bisection}
\end{figure}

\subsection{Top-Side Frenet Frame}

By Lemma~\ref{osculating plane bisects C^2},
at $C^2$ cone-free points $\vec X(s)$,
we can define the \emph{top-side normal} of the osculating plane
$\vechat B=\pm \vec B =  \pm \vec T\times \vec N$
whose sign is defined such that
$\vechat B \cdot \vec P_L = \vechat B \cdot \vec P_R > 0$.
Thus $\vechat B$ consistently points to the front side of the surface.
By contrast, $\vec B$'s orientation depends on whether the 2D curve locally
turns left or right (given by the sign of $k(s)$), flipping orientation at
inflection points (where $k(s) = 0$).

More formally, we will use the \emph{top-side Frenet frame}
given by $(\vec T(s), \vechat N(s), \vechat B(s))$ where
$\vechat N(s) = \vechat B(s) \times \vec T(s)$.

\begin{lemma} \label{top-side Frenet semikink}
  Consider a folded curved crease $\vec X(s)$
  that is cone-free at a semikink $s = \tilde s$.
  The top-side Frenet frames are identical in positive and negative limits:
  $$\lim_{s \to \tilde s^+} (\vec T(s), \vechat N(s), \vechat B(s))
  = \lim_{s \to \tilde s^-} (\vec T(s), \vechat N(s), \vechat B(s))$$
  and thus the top-side Frenet frame is continuous at $s = \tilde s$.
\end{lemma}

\begin{proof}
  First, $\vec T(\tilde s)$ is continuous because $\vec X(s)$ is
  $C^1$ at a semikink $s = \tilde s$.

  Second, by Lemma~\ref{osculating plane bisects C^2}, in the positive and
  negative limits, $\vec B(s)$ bisects $\vec P_L(s)$ and $\vec P_R(s)$.
  Because there is no cone ruling at $s = \tilde s$,
  the left and right surface normals $\vec P_L(s)$ and $\vec P_R(s)$
  have equal positive and negative limits at $\tilde s$,
  so $\vec P_L(\tilde s)$ and $\vec P_R(\tilde s)$ are continuous.
  Thus $\vec B(\tilde s^+)$ and $\vec B(\tilde s^-)$ must lie on a common
  bisecting line of $\vec P_L(\tilde s)$ and $\vec P_R(\tilde s)$,
  and $\vechat B(\tilde s)$ is uniquely defined by having positive dot product
  with $\vec P_1(\tilde s)$ and $\vec P_2(\tilde s)$.
  This gives us a unique definition of $\vechat B(s)$.

  Third, $\vechat N(s)$ is continuous as $\vechat B(s) \times \vec T(s)$.
  Therefore $(\vec T(s), \vechat N(s), \vechat B(s))$ is continuous at
  $s = \tilde s$.
\end{proof}


At $C^2$ points $X(s)$, we can define the \emph{signed curvature} $\hat K(s)$
to flip sign where $\vechat N(s)$ does:
$\hat K(s) \vechat N(s) = K(s) \vec N(s)$.
As in 2D, $\hat K(s)$ is positive where the curve turns left and
negative where the curve turns right (relative to the top side).



\subsection{General Bisection Property}

By combining Lemmas~\ref{osculating plane bisects C^2}
and~\ref{top-side Frenet semikink}, we obtain a stronger bisection lemma:

\begin{corollary} \label{osculating plane bisects}
  For a cone-free folded curved crease $\vec X(s)$,
  $\vechat B(s)$ bisects $\vec P_L(s)$ and $\vec P_R(s)$.
  In particular, the tangent planes of the surfaces on both sides
  of $\vec X(s)$ form the same angle with the osculating plane.
\end{corollary}

\subsection{Consequences}

Using the bisector property, we can prove the nonexistence of
a few strange situations.

\begin{lemma} \label{no planar}
  A crease $\vec X$ curved at $s$ cannot have a positive-length interval
  $s \in (s-\varepsilon, s+\varepsilon)$ incident to a planar region.
\end{lemma}

\begin{proof}
  If this situation were to happen, then the osculating plane of the curve
  must equal the plane of the planar region, which is say the left surface
  plane.  By Corollary~\ref{osculating plane bisects},
  the right surface plane must be the same plane.  But then the folded
  piece of paper is actually planar along the crease, contradicting that
  it is not $C^1$ along the crease.
\end{proof}

\begin{lemma} \label{no tangent}
  A rule segment cannot be tangent to a cone-free curved crease point 
  (at a relative interior point, in 2D or 3D).
\end{lemma}

\begin{proof}
  Suppose by symmetry that a rule segment is tangent to a crease point
  on its left side.
  If a rule segment is tangent to the crease point $\vec x(s)$ in 2D,
  then it must also be tangent to $\vec X(s)$ in 3D.
  There are two cases: (1) the left surface is a tangent
  surface generated from the crease; (2) the surface is trimmed by the crease
  and is only tangent at the point $\vec X(s)$.

  In Case~1, there is a finite portion of the crease that is $C^2$ and tangent
  to the incident rule segment.
  Then, for that portion of the crease (including~$s$),
  the tangent plane of the left surface is the osculating plane of the curve.
 
  In Case~2, consider surface normal $\vec P_L(s)$ at $\vec X(s)$.
  By assumption, the tangent vector $\vec T$ is parallel to the rule segment
  incident to $\vec X(s)$.
  Suppose by symmetry that $\vec T$ is actually the direction of the rule
  segment from $\vec X(s)$.
  (Otherwise, we could invert the parameterization of $\vec X$.)
  Because the surface normal is constant along the rule segment,
  and thus in the rule-segment direction, we have
  $${d \vec P_L \over d s^+} = \vec 0. $$
  Because $\vec P_L$ and $\vec T$ are perpendicular,
  ${d \over d s^+}(\vec P_L \cdot \vec T) = 0$, which expands to
  $$
  {d \vec P_L \over d s^+} \cdot \vec T +  \vec P_L \cdot {d \vec T \over d s^+} = 0.
  $$
  Thus we obtain $\vec P_L \cdot {d \vec T \over d s^+} = 0$.
  Because the folded crease is not straight (Lemma~\ref{no straight}),
  $\vec N$ is perpendicular to $\vec P_L$.
  Therefore the left tangent plane equals the osculating plane.

  By Corollary~\ref{osculating plane bisects}, in either case,
  the right tangent plane must also equal the osculating plane,
  meaning that the folded piece of paper is actually planar along the crease,
  contradicting that it is not $C^1$ along the crease.
\end{proof}


When the crease is $C^2$,
Lemma~\ref{no tangent} also implicitly follows from the Fuchs--Tabachnikov
relation between fold angle and rule-segment angle
\cite{Fuchs-Tabachnikov-1999, Fuchs-Tabachnikov-2007-developable}.

\begin{corollary} \label{rule lines exist}
  For a crease $\vec X$ curved and cone-free at $s$, the point $\vec X(s)$ has
  an incident positive-length rule segment on the left side of $\vec X$
  and an incident positive-length rule segment on the right side of $\vec X$.
\end{corollary}

\begin{proof}
  First, by Lemma~\ref{no planar}, $\vec X(s)$ is not locally surrounded
  by a flat region on either side, so by Corollary~\ref{uncreased region},
  $\vec X(s)$ must have a rule segment on its left and right sides.
  Furthermore, such a rule segment cannot be a zero-length limit of nearby
  rule segments, because such a rule segment would be tangent to the curve,
  contradicting Lemma~\ref{no tangent}.
\end{proof}

\xxx{above corollary should also extend to cone case...}


\begin{corollary}
  If a face's boundary is a $C^1$ curved closed curve,
  then the folded face's boundary is not $C^1$.
\end{corollary}

\begin{proof}
  Consider the decomposition from Corollary~\ref{uncreased region}
  applied to the face, resulting in planar and ruled regions.
  By Lemma~\ref{no planar}, the ruled regions' boundary collectively
  cover the face boundary.
  The planar regions form a laminar (noncrossing) family in the face,
  so there must be a ruled region adjacent to only one planar region
  (or zero if the entire folded face is ruled).
  This ruled region is either the entire folded face
  or bounded by a portion of the face boundary
  and by a single rule segment (bounding a planar region).
  For each rule segment in the ruled region,
  we can discard the side that (possibly) contains the boundary rule segment,
  effectively shrinking the rule region while preserving its boundary structure
  of partial face boundary and one rule segment.
  In the limit of this process, we obtain a rule segment that is tangent
  to the face boundary.
  By Lemma~\ref{no tangent}, this situation can happen only if the
  face is cone ruled at some point, which by Theorem~\ref{cone kink}
  implies that the folded face boundary is not $C^1$.
\end{proof}


\section{Smooth Folding}
\label{Smooth Folding}

A \emph{smoothly folded crease} is a folded crease that is $C^1$, i.e.,
kink-free.
In Corollary~\ref{no semikink} below, we will show that a
smoothly folded crease is furthermore $C^2$, i.e., it cannot have semikinks.
A \emph{smooth folding} of a crease pattern is a folding in which every
crease is smoothly folded.
In this section, we characterize smooth folding as cone-free.


\begin{theorem} \label{cone kink}
  If a folded crease $\vec X$ has a cone ruling at a point $\vec X(s)$,
  then $\vec X$ is kinked at~$s$.
\end{theorem}

\begin{proof}
  Assume by symmetry that $\vec X(s)$ has a cone ruling on the left side,
  say clockwise from rule vector $\vec R_1$ to rule vector $\vec R_2$.
  Because the unfolded crease $\vec x$ is $C^1$, it has a tangent vector
  $\vec t$, so the left side of $\vec x(s)$ is, to the first order,
  the cone clockwise from $-\vec t$ to~$\vec t$.
  Thus we have $-\vec t$, $\vec r_1$, $\vec r_2$, and $\vec t$ appearing in
  clockwise order around $\vec x(s)$, giving us the angle relation:
  $$ 180^\circ = \angle(-\vec t, \vec t) =
     \angle(-\vec t, \vec r_1) + \angle(\vec r_1, \vec r_2) +
     \angle(\vec r_2, \vec t). $$

  Now assume for contradiction that $\vec X$ is $C^1$ at~$s$,
  so we can define the tangent vector $\vec T(s)$.
  By triangle inequality on the sphere, we have
  $$180^\circ = \angle(-\vec T, \vec T) \leq
    \angle(-\vec T, \vec R_1) + \angle(\vec R_1, \vec R_2) +
    \angle(\vec R_2, \vec T). $$
  The latter three 3D angles must be smaller or equal to
  than the corresponding angles in 2D, by isometry.
  Furthermore, $\angle(R_1,R_2) < \angle(\vec r_1, \vec r_2)$,
  because the surface must be bent along the entire cone ruling
  (otherwise it would have a flat patch).
  Therefore
  $$\angle(-\vec T, \vec R_1) + \angle(\vec R_1, \vec R_2) +
    \angle(\vec R_2, \vec T) <
    \angle(-\vec t, \vec r_1) + \angle(\vec r_1, \vec r_2) +
     \angle(\vec r_2, \vec t) = 180^\circ, $$
  a contradiction.
\end{proof}

\begin{figure}
  \centering
  \includegraphics[width=\linewidth]{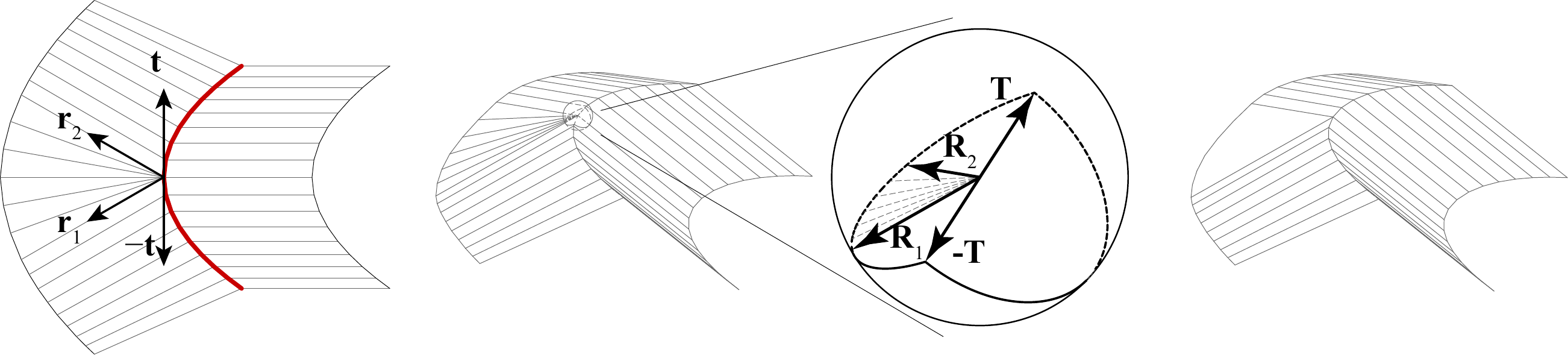}
  \caption{Cone rulings must fold into a kink in 3D.}
  \label{fig:cone-apex}
\end{figure}

Now we get a characterization of smooth folding:

\begin{corollary} \label{cone = kink}
  A folded curved crease $\vec X$ is kinked at $s$ if and only if it has
  a cone ruling at $\vec X(s)$.
\end{corollary}

\begin{proof}
  Theorem~\ref{cone kink} proves the ``if'' implication.

  To prove the converse, consider a cone-free crease point $\vec X(s)$.
  In 2D, we have a $180^\circ = \angle(-\vec t, \vec t)$ angle
  on either side of the crease.
  We claim that this $180^\circ$ angle between the backward tangent and
  forward tangent is preserved by the folding, so the folded
  crease $\vec X$ has a continuous tangent and thus is $C^1$ at~$s$.

  First, suppose that there is no planar region incident to $\vec X(s)$
  on say the left side.  Then the left side is locally a uniquely ruled $C^2$
  surface, with no rule segments tangent to the curve by Lemma~\ref{no tangent},
  and thus the surface can be extended slightly to include $\vec X(s)$
  in its interior.  In a $C^1$ surface, it is known\xxx{reference}
  that geodesic (2D) angles equal Euclidean (3D) angles, so folding
  preserves the $180^\circ$ angle between the backward and forward tangents.

  Now suppose that there is a planar region on the left side of $\vec X(s)$.
  By Lemma~\ref{only one planar}, there can be only one,
  and by Lemma~\ref{no planar}, there must be two uniquely ruled surfaces
  separating such a planar region from the crease.
  These three surfaces meet smoothly with a common surface normal, as the
  surface is $C^2$ away from the crease, so the overall angle between the
  backward and forward tangents of the crease equals the sum of the three
  angles of the surfaces at $\vec X(s)$.
  The previous paragraph argues that the two uniquely ruled surfaces
  preserve their angles, and
  the planar region clearly preserves its angle (it is not folded).
  Hence, again, folding preserves the $180^\circ$ angle between the backward
  and forward tangents.
\end{proof}


\section{Mountains and Valleys}
\label{Mountains and Valleys}

\subsection{Crease}
Refer to Figure~\ref{fig:bisection}.
For a smoothly folded (cone-free) crease $\vec X$,
the \emph{fold angle} $\rho \in (-180^\circ,180^\circ)$ at $\vec X(s)$
is defined by
$\cos \rho = \vec P_L \cdot \vec P_R$ and
$\sin \rho = [(\vec P_L \times \vec P_R) \cdot \vec T]$.
The crease is \emph{valley} at $s$ if the fold angle is negative, i.e.,
$(\vec P_L \times \vec P_R) \cdot \vec T < 0$. 
The crease is \emph{mountain} at $s$ if the fold angle is positive, i.e.,
$(\vec P_L \times \vec P_R) \cdot \vec T > 0$.

\begin{lemma} \label{fold angle nonzero}
  A smoothly folded curved crease $\vec X$ has a continuous fold angle
  $\rho \neq 0$.
\end{lemma}

\begin{proof}
  By Corollary~\ref{cone = kink}, the crease is cone-free,
  so the surface normals $\vec P_L(s)$ and $\vec P_R(s)$ are continuous.
  If the resulting fold angle $\rho(s)$ were zero, then
  we would have $\vec P_L(s) = \vec P_R(s)$, contradicting that
  the folded piece of paper is not $C^1$ at crease point $\vec X(s)$.
\end{proof}

\begin{corollary}
  A smoothly folded curved crease $\vec X$ is mountain or valley throughout.
\end{corollary}

\begin{proof}
  By Lemma~\ref{fold angle nonzero}, $\rho(s)$ is continuous and nonzero.
  By the intermediate value theorem, $\rho(s)$ cannot change sign.
\end{proof}

\begin{lemma}
\label{curvature in 2D and 3D}
  For a smoothly folded curved crease $\vec X(s)$,
  $$\hat K(s) \cos {1 \over 2} \rho(s) = \hat k(s).$$
  In particular, folding increases curvature:
  $|\hat k(s)| < |\hat K(s)|$, i.e., $k(s) < K(s)$.
\end{lemma}

\begin{proof}
  Referring to Figure~\ref{fig:bisection}, we have
  $$\cos {1 \over 2} \rho(s) = \vec P_L(s) \cdot \vechat B(s).$$
  By definition of $\vechat B(s)$, this dot product is the triple product
  $$\vec P_L(s) \cdot (\vec T(s) \times \vechat N(s)) =
    \vechat N(s) \cdot (\vec P_L(s) \times \vec T(s))$$
  (similar to the proof of Lemma~\ref{osculating plane bisects C^2}).
  Multiplying by $\hat K(s)$, we obtain
  $$(\hat K(s) \vechat N(s)) \cdot (\vec P_L(s) \times \vec T(s))
  = (K(s) N(s)) \cdot (\vec P_L(s) \times \vec T(s)).$$
  By Lemma~\ref{geodesic curvature}, this geodesic curvature is $\vechat k(s)$.
\end{proof}

\begin{corollary}
\label{no semikink}
  A folded crease cannot have a semikink, and thus
  a smoothly folded crease $\vec X$ is $C^2$.
\end{corollary}

\begin{proof}
  Suppose $\vec X(s)$ had a semikink at $s = \tilde s$.
  Applying Lemma~\ref{curvature in 2D and 3D} with positive and negative
  limits, we obtain that
  $$\lim_{s \to \tilde s^+} \hat K(s)
  = {\hat k(s) \over \cos{{1\over 2} \rho}}
  = \lim_{s \to \tilde s^-} \hat K(s),$$
  and thus the signed curvature $\hat K(s)$ is continuous at $s=\tilde s$.
  By Lemma~\ref{top-side Frenet semikink}, $\vechat N(s)$ is continuous
  at $s=\tilde s$.
  Therefore ${d^2 \vec X(s) \over {ds}^2} = \hat K(s) \vechat N(s)$ is
  continuous at $s = \tilde s$,
  so $\vec X(\tilde s)$ is not actually a semikink.
\end{proof}



\begin{lemma}
  A smoothly folded crease $\vec X$ is valley if and only if
  $(\vec P_L \times \vechat B) \cdot \vec T < 0$,
  and mountain if and only if
  $(\vec P_L \times \vechat B) \cdot \vec T > 0$.
\end{lemma}

\begin{proof}
  Refer to Figure~\ref{fig:bisection}.
  Vectors $\vec P_L$, $\vec P_R$, and $\vechat B$ are all perpendicular
  to~$\vec T$, and thus live in a common oriented plane with normal $\vec T$.
  By the choice of $\vechat B$ to have positive dot products with $\vec P_L$ and
  $\vec P_R$, the three vectors in fact live in a common half-plane.
  In this plane, we can see the fold angle
  $\rho = \angle (\vec P_L, \vec P_R)$, where $\angle$ measures the convex
  angle between the vectors, signed positive when the angle is convex in the
  counterclockwise orientation within the oriented plane with normal $\vec T$,
  and signed negative when clockwise.

  By Corollary~\ref{osculating plane bisects},
  $\vec P_L \cdot \vec B = \vec P_R \cdot \vec B$, so
  $\vec P_L \cdot \vechat B = \vec P_R \cdot \vechat B$.
  Thus $\cos \angle (\vec P_L, \vechat B) = \cos \angle (\vec P_R, \vechat B)$, i.e.,
  $|\angle (\vec P_L, \vechat B)| = |\angle (\vec P_R, \vechat B)|$.
  
  If $\angle (\vec P_L, \vechat B) = \angle (\vec P_R, \vechat B)$,
  then $\vec P_L = \vec P_R$, contradicting that $\vec X$ is a crease.
  Therefore, $\angle (\vec P_L, \vechat B) = \angle (\vechat B, \vec P_R)
  = \pm {1 \over 2} \angle (\vec P_L, \vec P_R)$.
  Because $|\angle (\vec P_L, \vechat B)| < 90^\circ$, we must in fact have
  $\angle (\vec P_L, \vechat B) = \angle (\vechat B, \vec P_R)
  = {1 \over 2} \angle (\vec P_L, \vec P_R)$,
  i.e., $\vechat B$ bisects the convex angle $\angle (\vec P_L, \vec P_R)$.
  Hence $\vechat B$ lies in between $\vec P_L$ and $\vec P_R$
  within the half-plane.
  Therefore the cross products $\vec P_L \times \vec P_R$,
  $\vec P_L \times \vechat B$, and $\vechat B \times \vec P_R$ are all parallel,
  so their dot products with $\vec T$ have the same sign.
\end{proof}

\subsection{Rule Segment}

We can also define whether a rule segment bends the paper mountain or valley;
refer to Figure~\ref{fig:rulelineMV}.
Consider a relative interior point $\vec Y$ of a rule segment with direction
vector $\vec R$, with top-side surface normal $\vec P$.
Then we can construct a local Frenet frame at $\vec Y$
with tangent vector $\vec Q = \vec R \times \vec P$,
normal vector $\vec P$, and binormal vector $\vec R$.
These frames define a 3D curve $\vec Y(t)$ where $\vec Y(0) = \vec Y$,
which follows the principle curvature of the surface.
Parameterize this curve by arc length.

\begin{figure}
  \centering
  \includegraphics[width=0.8\linewidth]{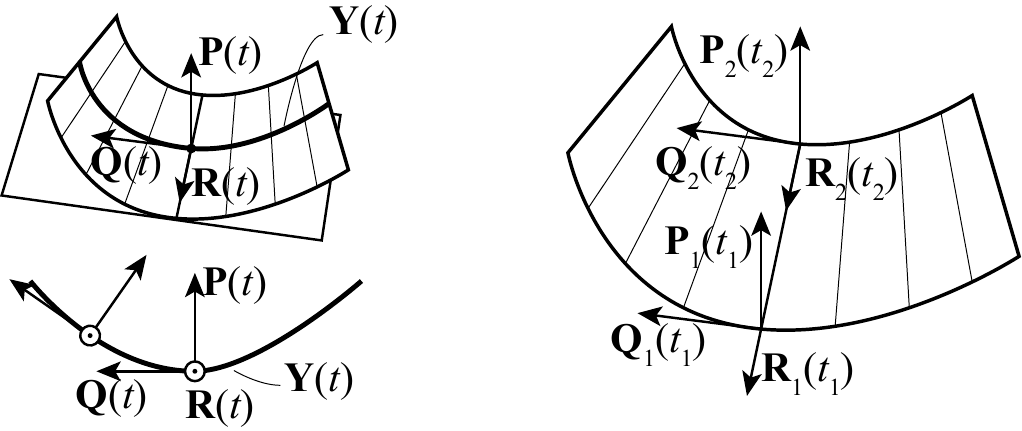}
  \caption{Defining a frame around an interior point to define
           mountain vs.\ valley bending.}
  \label{fig:rulelineMV}
\end{figure}

First consider the case when the surface is $C^2$ at $Y(t)$.
The surface \emph{bends valley} at $\vec Y(t)$ if the curvature vector
${d^2 \vec Y(t) \over d t^2} = {d \vec Q(t) \over d t}$ is on the top side,
i.e., has positive dot product with $\vec P(t)$; and it \emph{bends mountain}
if ${d \vec Q(t) \over d t} \cdot \vec P(t) < 0$.
In particular, at $t=0$, we determine whether the original rule segment
bends mountain or valley at $\vec Y$.

If the surface is not $C^2$ at $Y(t)$, then the rule segment is a semicrease,
which connects two $C^2$ surfaces sharing a surface normal at the crease;
refer to Figure~\ref{fig:semicreaseMV}.
In this case, the surface bends valley at $\vec Y(t)$ when the two surfaces
bends valley; or one of the surface is planar, and the other bends valley.
Similarly, the surface bends mountain at $\vec Y(t)$ when the two surfaces
bends mountain; or one of the surface is planar, and the other bends mountain.
At an inflection point, there is no mountain/valley assignment.

\begin{figure}
  \centering
  \includegraphics[width=1.0\linewidth]{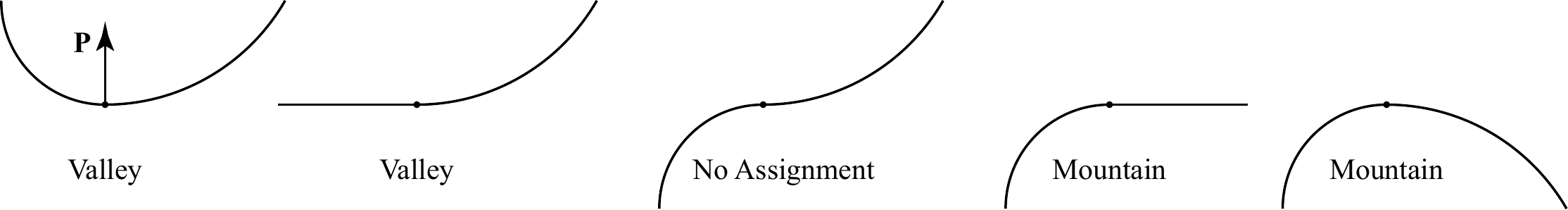}
  \caption{Definition of mountain and valley for a semicrease.}
  \label{fig:semicreaseMV}
\end{figure}


\begin{lemma}
\label{bends same direction}
  A developable uncreased surface bends the same direction (mountain or valley)
  at every relative interior point of a rule segment.
\end{lemma}

\begin{proof}
  First consider the case when the surface is $C^2$.
  Consider two points $\vec Y_1$ and $\vec Y_2$ on the rule segment, with
  principle curvature frames $(\vec Q_i(t_i), \vec R_i(t_i), \vec P_i(t_i))$.
  Choose $t_2$ as a function of $t_1$ so that $\vec Y_1(t_1)$ and
  $\vec Y_2(t_2)$ lie on a common rule segment.
  Then the frames are in fact identical:
  $\vec R_1(t_1) = \vec R_2(t_2)$ is the common rule direction,
  $\vec P_1(t_1) = \vec P_2(t_2)$ is the common top-side surface normal,
  and $\vec Q_1(t_1) = \vec Q_2(t_2)$ is their cross product.
  Because the surface is locally $C^2$ around the ruled segment $\vec Y_1$ and $\vec Y_2$,
  we have ${d t_2 \over d t_1} > 0$, so
  $$ {d \vec Q_2(t_2) \over d t_2} \cdot \vec P = 
      {d \vec Q_1(t_1) \over d t_2} \cdot \vec P = 
      {d t_2 \over d t_1}{d\vec  Q_1(t_1) \over d t_1} \cdot \vec P.
  $$
  Therefore, the surface bends the same direction.

  Next consider the case when the surface is not $C^2$, i.e., the rule segment is a semicrease between $C^2$ surfaces $S^+$ and $S^-$.
  By the above argument, in a $C^2$ patch, the inflection occurs along the rule segment where ${d \vec Q(t) \over d t} \cdot \vec P = 0$ is satisfied.
  Also, if surface is not $C^2$, then it is on a rule segment.
  Therefore, if the $S^-$ surface is bent in a different direction at $\lim_{t\rightarrow t_1^-}\vec Y_1(t)$ and $\lim_{t\rightarrow t_2^-}\vec Y_2(t_2)$, then a path from $\vec Y_1$ to $\vec Y_2$ must cross a rule segment.
  Because rule segments do not intersect\xxx{cite}, $S^+$ and $S^-$ keeps their own bending orientations.
  Therefore, the assignment for the semicrease is unchanged along the segment.
\end{proof}

By Lemma~\ref{bends same direction}, we can define the bending direction of a
rule segment: a developable uncreased surface bends mountain or valley at a
rule segment if a relative interior point of the rule segment bends mountain
or valley, respectively.  Furthermore, because the frames are identical, we
can define the principle curvature frame $(\vec Q, \vec R, \vec P)$ of a rule
segment by the principle curvature frame at any relative interior point on the
rule segment.

\subsection{Crease vs.\ Rule Segment}

Next we consider the mountain-valley relation between a rule segment and a
crease.


First consider a smoothly folded crease $\vec X$ with left and right surface
\emph{ruling vectors} $\vec R_L$ and $\vec R_R$,
defined as unit vectors which lie along the rule
segments on surfaces $S_L$ and $S_R$ incident to~$\vec X$.
(If there is a planar region incident to $\vec X$,
 these ruling vectors will not be unique.)
A left-side ruling vector $\vec R_L$ lives in the plane perpendicular to $\vec P_L$.
Therefore, the vector can be represented by $$\vec R_L = (\cos \theta_L) \vec T + (\sin \theta_L) (\vec P_L \times \vec T),$$ where we call $\theta_L$ the \emph{left-side ruling angle} of the ruling, which is nonzero by Lemma~\ref{no tangent}.
Because the ruling angle is intrinsic, the ruling vector in 2D is represented
by $\vec r_L = (\cos \theta_L) \vec t + (\sin \theta_L) \vechat b$.
The orientation of the left-side ruling vector is chosen to orient to the left, i.e., $\vec r_L \cdot \vechat b > 0$, so $\theta_L$ is positive.
Similarly, ruling vector  $\vec R_R$ on the right surface is represented by $\vec R_R = (\cos \theta_R) \vec T - (\sin \theta_R) (\vec P_R \times \vec T)$, using \emph{right-side ruling angle} $\theta_R$.
The orientation is chosen to be on the right side, so $\theta_R>0$.

\begin{lemma}
  Consider a uniquely ruled smoothly folded crease $\vec X$
  with locally $C^2$ surfaces on both sides (no semicreases).
  Then the rule segment on the left side of $\vec X$ bends valley
  if and only if $\vec N \cdot \vec P_L > 0$.
  Symmetrically, the surface bends valley on the right side if and only if $\vec N \cdot \vec P_R > 0$.
\end{lemma}

\begin{proof}
  Build the principle curvature frame $(\vec Q(t), \vec R(t), \vec P(t))$ of rule segment
  parameterized by the arclength $t$ in the principle curvature direction.
  Consider corresponding point $\vec X(s)$ and the arclength parameter $s=s(t)$
  along the crease at the rule segment parameterized by $t$.
  Because the surface is locally $C^2$ around the rule segment, ${d s \over d t} > 0$.
  Because we consider the left side of the surface, $\vec P_L(s) = \vec P(t)$.
  Let $\theta$ be the angle between $\vec R(t)$ and $\vec T(s)$,
  i.e., $\vec T(s) = \sin \theta \vec Q(t) + \cos \theta \vec R(t)$.
  By Lemma~\ref{no tangent}, $0<\theta<\pi$, and we get 
  \begin{equation}
  \vec Q = (\csc \theta) \vec T - (\cot \theta) \vec R.
  \label{eq: Q by T and R}
  \end{equation}

  Assume that the surface bends
  valley at the rule segment, i.e.,
  \begin{equation}
  V (t) = {d \vec Q(t) \over d t} \cdot \vec P(t) > 0.
  \label{eq: bends valley}
  \end{equation}
  Using orthogonality of vectors $\vec Q$ and $\vec P$, i.e., $\vec Q(t)\cdot \vec P(t) = 0$, and taking derivatives, we obtain
  $$
  {d \vec Q \over d t} \cdot \vec P  +  \vec Q \cdot {d \vec P \over d t} = 0.
  $$
  Then, 
  \begin{eqnarray*}
  V(t) &=& -\vec Q \cdot {d \vec P\over d t} \\
        &=& - \Big ((\csc \theta) \vec T - (\cot \theta) \vec R \Big) \cdot {d \vec P \over d t}\\
        &=& - (\csc \theta) \vec T \cdot {d \vec P \over d t}.
  \end{eqnarray*}
  Here, we used Equation~\ref{eq: Q by T and R}.
  By the orthogonality of vectors $\vec T$ and $\vec P$, we get 
  $$ \vec T \cdot {d \vec P \over d t}=   {d \vec T \over d t} \cdot \vec P. $$
  Then,
  \begin{eqnarray*}
   V(t) &=& (\csc \theta) {d \vec T(s) \over d t} \cdot \vec P(t) \\
         &=& (\csc \theta) {d s \over d t}{d \vec T(s) \over d s}\cdot \vec P(t) \\
         &=& (\csc \theta) {d s \over d t} K(s) \vec N(s) \cdot \vec P_L(s).
  \end{eqnarray*}
  Because $\csc \theta > 0$, ${d s \over d t}>0$, and $K(s) >0 $, Equation~\ref{eq: bends valley} is equivalent to
  $
  \vec N(s) \cdot \vec P_L(s) > 0.
  $
\end{proof}

Now we make a stronger statement, allowing the ruling vectors to be not unique
and the surfaces to be not $C^2$.

\begin{corollary}
  Consider a smoothly folded crease $\vec X$.
  Then a rule segment on the left side of $\vec X$ bends valley
  if and only if $\vec N \cdot \vec P_L > 0$.
  Symmetrically, the surface bends valley on the right side if and only if $\vec N \cdot \vec P_R > 0$.
\label{crease vs ruling MV}
\end{corollary}

\begin{proof}
  Consider rule segments at $\vec X(\tilde s)$.
  By Theorem~\ref{cone kink}, the crease is cone free,
  so a rule segment is either (1)~between two $C^2$ ruled surfaces,
  or (2)~between a plane and a $C^2$ ruled surfaces.
  
  Consider Case~1, and let $S^-$ and $S^+$ be the two surfaces.
  Because there are no cone rulings, $S^-$ and $S^+$ are locally formed by
  unique rulings emanating from $\vec X(s)$ at $s<\tilde s$ and $s>\tilde s$,
  respectively.
  Then
  $$
  \lim_{s \to \tilde s^-} \vec N(s) \cdot \vec P_L(s) =
  \lim_{s \to \tilde s^+} \vec N(s) \cdot \vec P_L(s) =
  \vec N(s) \cdot \vec P_L(s).
  $$
  So both surfaces $S^-$ and $S^+$ bend valley if and only if $\vec N(s) \cdot \vec P_L(s) >0$.

  Next consider Case~2.
  By symmetry, assume that $S^-$ is planar and $S^+$ is $C^2$ ruled surface.
  Then $S^+$ is locally formed by unique rule segments emanating from
  $\vec X(s)$ at $s > \tilde s$.
  Hence $S^+$, and thus the rule segment, bends valley if and only if $\vec N(s) \cdot \vec P_L(s) >0$.
\end{proof}


\begin{theorem}
  Consider a smoothly folded curved crease $\vec X$.
  A rule segment incident to $\vec X(\tilde s)$ on the convex side of $\vec X(\tilde s)$
  has the same mountain/valley assignment as the crease, while
  a rule segment incident to $\vec X(\tilde s)$ on the concave side of $\vec X(\tilde s)$
  has the opposite mountain/valley assignment as the crease.
\label{convex side bends same direction}
\end{theorem}
\begin{proof}
Assume by symmetry that the left side of the paper is the convex side ($\hat k(s) < 0$).
Also, assume that the crease is a valley, i.e., $(\vechat B \times \vec P_L) \cdot \vec T =( \vec P_L \times \vec  B) \cdot \vec T> 0$.
Then the top-side normal of the osculating plane is $\vechat B = -\vec B$, and thus $\vechat N = -\vec N$.

Now
\begin{eqnarray*}
(\vec P_L \times \vec  B) \cdot \vec T &= & 
\left(\vec P_L \times (\vec T \times \vec N)\right) \cdot \vec T \\
&= & \left(\vec T  (\vec P_L \cdot \vec N) - \vec N(\vec P_L \cdot \vec T)\right)\cdot \vec T > 0.
\end{eqnarray*}
The second term disappears because $\vec P_L \cdot \vec T =0$. Therefore
$
\vec P_L \cdot \vec N > 0
$,
so the left side is valley.
\end{proof}

\subsection{Creases Connected by a Rule Segment}




Now consider two creases connected by a rule segment.
By Lemma~\ref{convex side bends same direction}, we get the following.

\begin{corollary} \label{concave MV match}
  Consider two smoothly folded creases connected by a rule segment.
  If the rule segment is on the concave sides of both creases,
  or on the convex sides of both creases,
  then the creases must have the same direction (mountain or valley).
  If a rule segment is on the convex side of one crease and the
  concave side of the other crease, then the creases must have
  the opposite direction (one mountain and one valley).
\end{corollary}


\section{Lens Tessellation}
\label{Lens Tessellation}


In this section, we use the qualitative properties of rulings obtained
in previous sections to reconstruct rule segments from a crease pattern
of the generalized version of lens tessellation.

First, as illustrated in Figure~\ref{fig:lens-rulings},
we define the \emph{lens tessellation} parameterized by a convex $C^2$ function
$\ell : [0,1] \to [0, \infty)$ with $\ell(0) = \ell(1) = 0$,
horizontal offset $u \in [0,1)$, and vertical offset $v \in (0,\infty)$,
to consist of
\begin{enumerate}
\item mountain creases
  $\gamma^\pm_{i, 2j} = \{(t + i, \pm \ell(t) +  j v) \mid t \in [0,1]\}$
  for $i, j \in \mathbb Z$; and
\item valley creases
  $\gamma^\pm_{i, 2j+1} = \{(1-t + i + u, \pm \ell(1-t) + (j + {1 \over 2} ) v)\} \mid t \in [0,1]\}$
  for $i, j \in \mathbb Z$.
\end{enumerate}
Define the \emph{vertices} to be points of the form
$V_{i, 2j} =(i, j v)$ and $V_{i, 2j+1}=(i + u, (j+{1 \over 2}) v)$.
Four creases meet at each vertex.

Because $\ell(t)$ is convex, it has a unique maximum $\ell(t^*)$
at some $t = t^*$.
Define the \term{apex} $A_{i,k}$ of crease $\gamma^\pm_{i,k}$
to be the point of the crease at $t=t^*$, i.e.,
$A^\pm_{i,2j} = (t^* + i, \pm \ell(t^*) + j v)$ and
$A^\pm_{i,2j+1}= (1-t^* + i + u, \pm \ell(1-t^*) + (j + {1 \over 2}) v)$.


\begin{figure}[tb]
  \centering
  \includegraphics[width=1\linewidth]{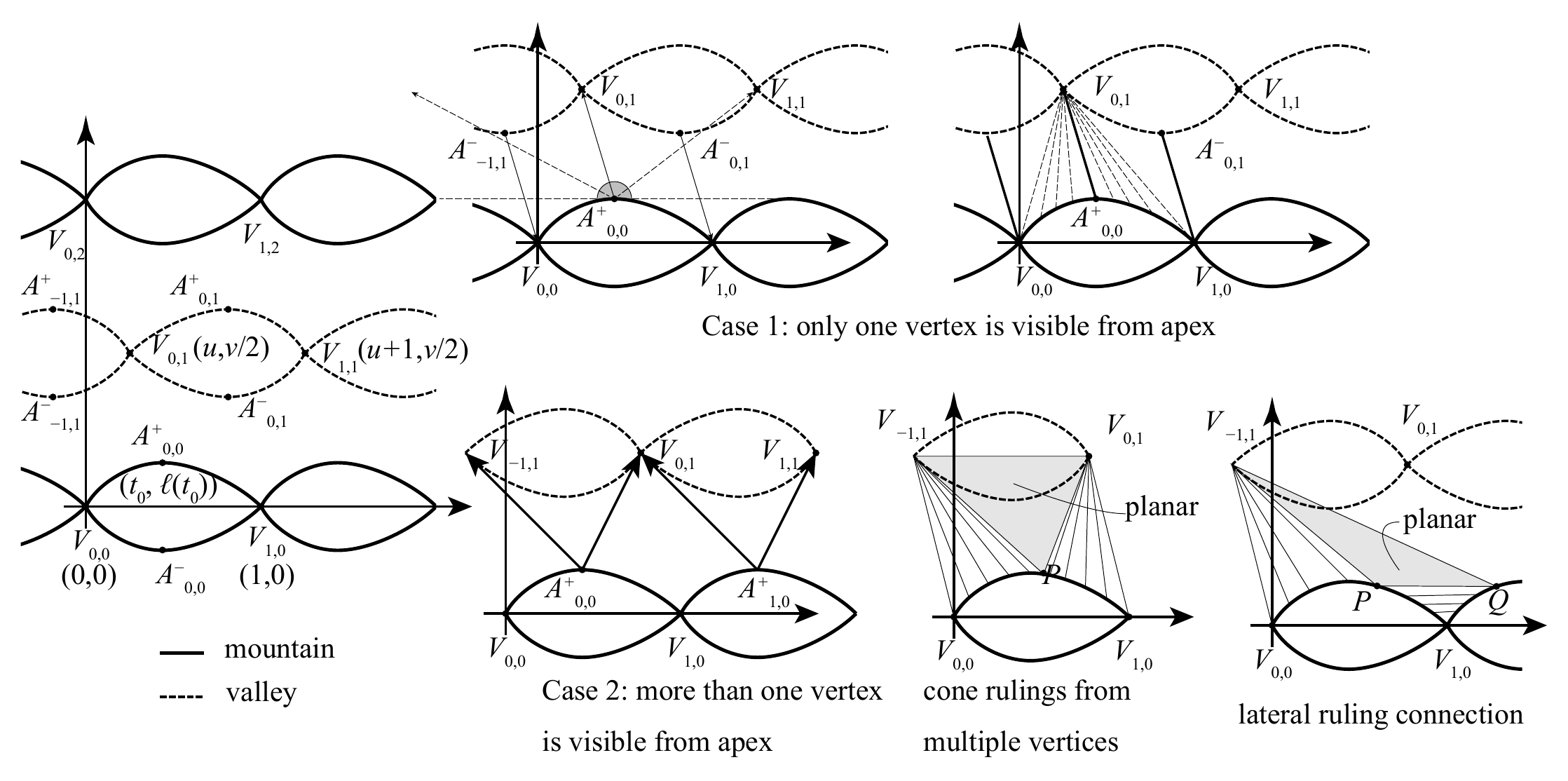}
  \caption{Ruling conditions for a lens tessellation.}
 \label{fig:lens-rulings}
\end{figure}

\subsection{Necessary Conditions}

Consider a crease point $\vec x(s)$.
A point $\vec y$ on the crease pattern (a vertex or crease point) is
\term{visible} from $\vec x(s)$ on the left (right) side of $\vec x$ at
$\vec x(s)$ if the oriented open line segment
$\overrightarrow{\vec x(s) \vec y}$ is on the
left (right) side of $\vec x(s)$ and does not share a point with the crease
pattern.
If $\vec x(s)$ and $\vec y$ are the endpoints of a rule segment,
then certainly they must be visible from each other.

\begin{theorem} \label{lens necessary}
  A lens tessellation can smoothly fold only if there is a vertex
  $V_{i,1}$ visible from every point on crease $\gamma^+_{0,0}$
  on the convex side.
\end{theorem}

\begin{proof}
Refer to Figure~\ref{fig:lens-rulings}.
By Corollary~\ref{rule lines exist}, there must be a rule segment emanating
from $A^+_{0,0}$ on the convex side of $\gamma^+_{0,0}$.
The other endpoint $B$ of that rule segment must be visible from $A^+_{0,0}$
on the convex side of $\gamma^+_{0,0}$.
Because the tangent line of $\gamma^+_{0,0}$ at $A^+_{0,0}$ is horizontal,
any such visible point $B$ must lie on the union of creases $\gamma^-_{i,1}$
and vertices $V_{i,1}$ for $i \in \mathbb Z$.
By Theorem~\ref{concave MV match}, $B$ cannot be on the relative interior
of one of the valley creases $\gamma^-_{i,1}$ because then the rule segment
would be on the concave sides of creases of opposite direction.
Thus $B$ must be among the vertices $V_{i,1}$ for $i \in \mathbb Z$.

First consider the case that only one vertex $V_{n,1}$ is visible from
$A^+_{0,0}$ on the convex side of $\gamma^+_{0,0}$.
Then $A^+_{0,0} V_{n,1}$ must be a rule segment.
By symmetry, $V_{1,0} A^-_{n,1}$ is also a rule segment.
Consider a point on $\gamma^+_{0,0}$ between $A^+_{0,0}$ and $V_{0,1}$,
which by Corollary~\ref{rule lines exist} has a rule segment on the
positive side of $\gamma^+_{0,0}$.
This rule segment cannot cross the existing rule segments
$A^+_{0,0} V_{n,1}$ and $V_{1,0} A^-_{n,1}$, so its other endpoint must be
$V_{n,1}$, $A^-_{n,1}$, or
between $V_{n,1}$ and $A^-_{n,1}$ on curve $\gamma^{-}_{n,1}$.
By Theorem~\ref{concave MV match}, the only possible ruling is to have a cone
apex at $V_{n,1}$.
Similarly, rule segments from points between $V_{0,0}$ and $A^+_{0,0}$ on
$\gamma^+_{0,0}$ must end at $V_{n,1}$.
Therefore $V_{n,1}$ is visible from every point on crease $\gamma^+_{0,0}$
on the convex side.

Second consider the case in which more than one vertex $V_{i,1}$ is visible
from apex $A^+_{0,0}$ on the convex side of $\gamma^+_{0,0}$.
Suppose for contradiction that there is no common vertex visible
from the entire curve $\gamma^+_{0,0}$.
Similar to the previous case, there must be a rule segment from apex
$A^+_{0,0}$ to one of the vertices $V_{n,1}$.
But we assumed that some other point of $\gamma^+_{0,0}$ cannot see $V_{n,1}$.
By symmetry, suppose that point is to the right of $A^+_{0,0}$.
There is a transition point $P$ on the relative interior of $\gamma^+_{0,0}$
when the endpoints of rulings change from $V_{n,1}$ to either
(a)~another vertex $V_{m,1}$ with $m > n$ or (b)~a point on $\gamma^+_{1,0}$.
(See Figure~\ref{fig:lens-rulings}.)
At such a point $P$, we have two rule segments.
By Theorem~\ref{cone kink}, $P$ cannot be a cone apex.
Hence there must be a planar region between the two rule segments.
Specifically, in case (a), the triangle $P V_{n,1} V_{m,1}$ is planar,
which contains all of $\gamma^-_{n,1}$, contradicting that the folded
piece of paper is not $C^1$ on $\gamma^-_{n,1}$.
In case (b), let $Q$ be the point on $\gamma^+_{1,0}$.
The triangle $P Q V_{n,1}$ is planar.
This triangle cannot intersect $\gamma^-_{n,1}$, because the folded piece
of paper is not $C^1$ on $\gamma^-_{n,1}$.
In particular, the curve $\gamma^-_{n,1}$ cannot intersect the segment
$V_{n,1} V_{0,1}$ (which begins in the triangle).
Because $\gamma^+_{1,0}$ is a $180^\circ$ rotation of $\gamma^-_{n,1}$ mapping
$V_{n,1}$ to $V_{0,1}$, we symmetrically have that the curve $\gamma^+_{1,0}$
cannot intersect the same segment $V_{n,1} V_{0,1}$.  Thus this segment is
a visibility segment, as is $V_{n,1} V_{0,0}$.  By convexity of the lens,
$V_{n,1}$ can see the entire curve $\gamma^+_{1,0}$.
Therefore there is in fact a common vertex visible from the curve $\gamma^+_{0,0}$.
\end{proof}

%

\subsection{Existence / Sufficiency}

Finally we prove that the condition from Theorem~\ref{lens necessary}
is also sufficient:

\begin{theorem} \label{lens sufficient}
  A lens tessellation can fold smoothly if
  there is a vertex $V_{i,1}$ visible from
  every point on crease $\gamma^+_{0,0}$ on the convex side.
\end{theorem}


\begin{proof}
First we construct the folding of one ``gadget'', $(i,j)=0$;
refer to Figure~\ref{fig:kite-module}.
We can add an integer to $u$ to assume that
$V_{0,1}$ is the visible vertex from apex $A^+_{0,0}$.
In 2D, this gadget is bounded by a quadrangle of rule segments
with vertex coordinates
$V_{0,0}=(0,0)$, $V_{0,-1} = (u, -{1 \over 2} v)$, $V_{1,0} = (1,0)$, $V_{0,1} = (u, {1 \over 2} v)$.
This kite module is decomposed by its creases into an upper wing part $U$, middle lens part $M$, and lower wing part $L$.
We assume that $M$ is ruled parallel to $y$ axis:
the rule segments of $M$ are of the form
$(t, \ell(t))$ and $(t, -\ell(t))$ parameterized by $t$.
(We can make this assumption because we are constructing a folded state.)
We also assume that $U$ consists of cone rulings between $V_{0,1}$ and $(t, \ell(t))$ while $L$ consists of cone rulings between $V_{0,-1}$ and $(t, -\ell(t))$ using the same parameter~$t$.

\begin{figure}[tb]
  \centering
  \includegraphics[width=0.8\linewidth]{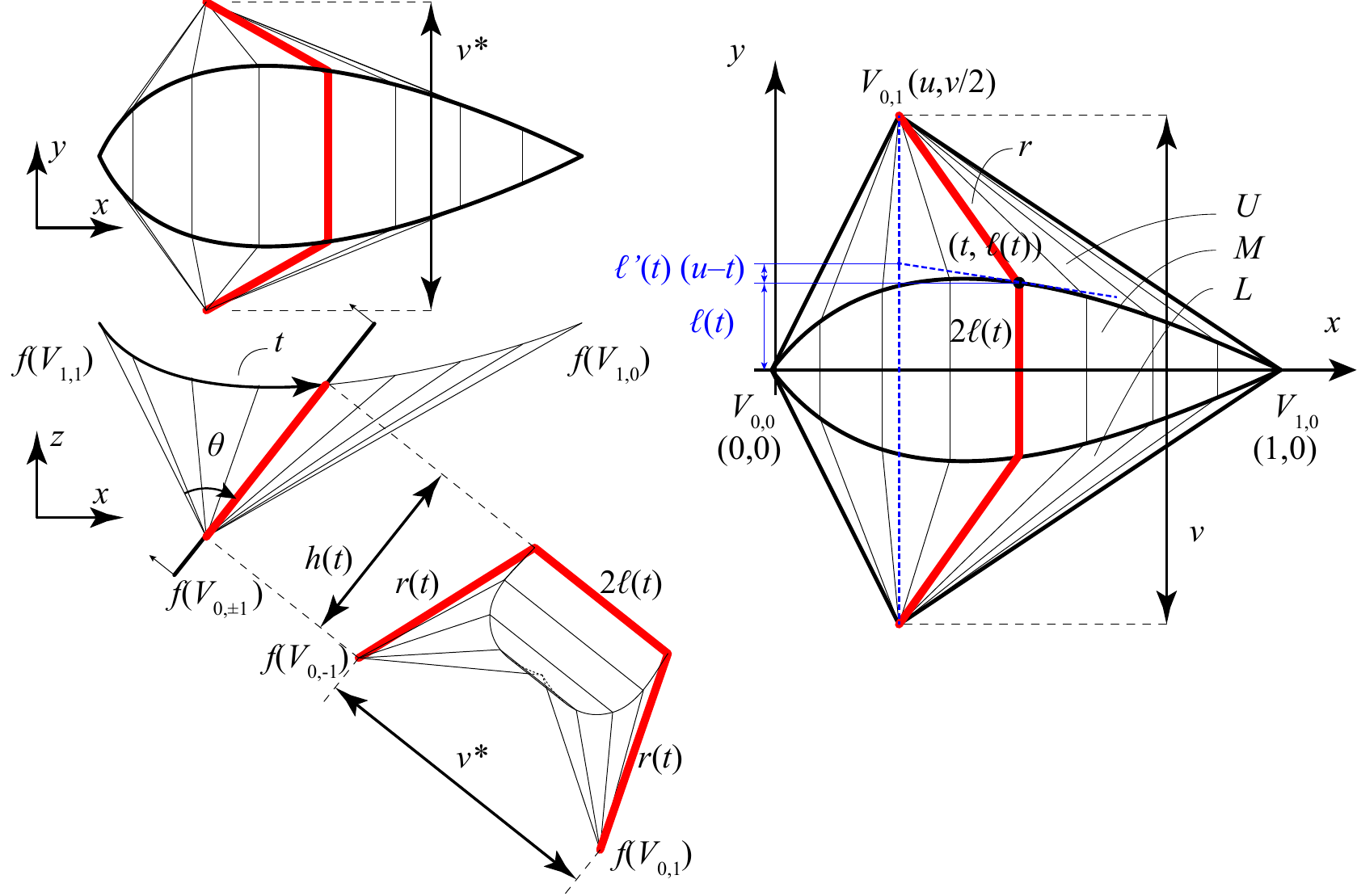}
  \caption{A modular kite structure.}
  \label{fig:kite-module}
\end{figure}

The folding $f(M)$ is a cylindrical surface with parallel rulings. 
We orient the folded form such that this ruling direction is parallel to
$y$ axis and $\overrightarrow{f(V_{0,0})f(V_{1,0})}$ is parallel
to the positive direction of $x$ axis.
Then the orthogonal projection of $f(M)$ to $x z$ plane is a curve $\gamma$,
and a ruling at $t$ on $M$ corresponds to a point on $\gamma(t)$
while $t$ is the arclength parameter.

We further assume that the folded state is symmetric with respect to reflection
through a plane passing through $ \overrightarrow{f(V_{0,0})f(V_{1,0})}$ and
parallel to $x z$ plane.
Let the distance between $f(V_{0,-1})$ and $f(V_{0,1})$ be denoted by $v^*$ where $0<v^*<v$.
We will show that there is a valid folded state for arbitrary $v^*$ if it is sufficiently close to~$v$.

Consider the set of rule segments of $U$, $M$, and $L$ at parameter $t$ and its folding.
Then, by our symmetry assumption, these segments form a planar polyline which
together with segments $f(V_{0,-1})$ and $f(V_{0,1})$ form an
isosceles trapezoid with the base length $v^*$ and top length $2\ell(t)$.
The legs are the length of the rule segments, which can be calculated from the
crease pattern as $r(t) = \sqrt{(u-t)^2 + (v/2-\ell(t))^2}$.
Such a trapezoid exists because $0<v^*<v\leq2\ell(t)+2r(t)$.
The height of the trapezoid $h(t)$ is given by 
$$
h(t) = \sqrt{(v-v^*) \left({{v+v^*} \over 4}-\ell(t)\right)+(t-u)^2}.
$$

Now consider the projection of this trapezoid in $x z$ plane.
This projection is a line segment between two points,
namely the projections of $V_{0,1}$ and $\gamma(t)$,
and it must have length of $h(t)$.
We use the following lemma to solve for~$\gamma$:

\begin{lemma}\label{ruled-surface-isometry}
  If an arc-length parameterized crease $\vec x(s)$
  have unique rule segments on one side incident to cone apex $\vec a$,
  then an embedding $\vec f$ is a proper folding if and only if
  folded curve $\vec X = \vec f \circ \vec x$ is also arc-length parameterized, and
  rule segments from $\vec a$ to $\vec x(s)$ maps isometrically to rule segments from $\vec A$ to $\vec X(s)$, 
  where $\vec A = \vec f \circ \vec a$.
\end{lemma}
\begin{proof}
 Necessity (only if part) is obvious, so we prove the sufficiency (if part).
 The folded curve is arc-length parameterized by $s$ as  ${d \vec X(s) \over ds} = {d \vec x(s) \over ds} = 1$, 
 and the length of ruling segment $L(s)$ must be equal $L(s) = \|\vec x(s) - \vec a\| = \|\vec X(s) - \vec A\|$.
 Let $\vec r(s)$ denote the unit ruling vectors from the apex toward the curve in 2D, 
 i.e., $\vec r(s) = (\vec x(s) - \vec a)/L(s)$.
 Similarly denote unit ruling vector in 3D by $\vec R(s) = (\vec X(s) - \vec A)/L(s)$.
 Consider a coordinate system using arc length $s$ and radius $\ell$.
 The conical portion of the face formed by the crease and a point 
 is uniquely ruled at any point, so $(s, \ell)$ uniquely represent a point on the portion.
 A point $(s, \ell)$ in 2D corresponds to $\vec a + \ell \vec r(s)$,
 which is mapped to 3D to $\vec A + \ell \vec R(s)$.
 Consider a 2D $C^1$ curve $\vec y(t)$ represented by $(s(t), \ell(t))$, where $t$ is the arc-length parameterization.
 Then the total derivative of $\vec y(t) = \vec a + \ell(t)  \vec r(s(t)) $ is
 $$
  {d \vec y \over dt}
  = {\partial \vec y \over \partial s}{ds\over dt} + {\partial \vec y \over \partial \ell} {d \ell\over dt} 
  =  \ell {d \vec r \over ds}{ds\over dt} +  \vec r {d\ell\over dt}.
 $$
 Then, by taking the dot product with itself,
 \begin{eqnarray*}
  \left\|{d \vec y \over dt}\right\|^2 &=&
  \ell^2 \left\|{d \vec r \over ds}\right\|^2 \left({ds \over dt}\right)^2 
  +  2 \ell {d \vec r \over ds} \cdot \vec r \left({ds \over dt}\right) \left({d\ell \over dt}\right) + \|\vec r\|^2 \left({d\ell \over dt}\right)^2\\
  &=& \ell^2 \left\|{d \vec r \over ds}\right\|^2 \left({ds \over dt}\right)^2 + \left({d\ell \over dt}\right)^2,
 \end{eqnarray*}
 where we used $\vec r \cdot \vec r= 1$ and $ 2{d \vec r \over ds} \cdot \vec r = {d \over ds}\left(\vec r \cdot \vec r \right) = 0$. 
 Because $L(s) \vec r(s) = \vec x(s) - \vec a$,
 taking derivatives yields
 $$
 L{d \vec r \over ds} + {d L \over ds}\vec r = {d \vec x \over ds}.
 $$
 By taking the dot product,
 $$
  L^2 \left\| {d \vec r \over ds} \right\|^2 + \left( {d L \over ds} \right)^2 = \left\|{d \vec x \over ds}\right\|^2 = 1,
 $$
 again using $ {d \vec r \over ds} \cdot \vec r = 0$ and $\vec r \cdot \vec r = 1$.
 Thus,
 $$
  \left\|{d \vec y \over dt}\right\|^2 = 
  {\ell^2\over L^2} \left( 1 -  \left( {d L \over ds} \right)^2 \right) \left({ds \over dt}\right)^2 + \left({d\ell \over dt}\right)^2.
 $$
 The mapped crease $\vec Y(t)$ in 3D is defined by $\vec Y(t) = \vec A + \ell(t)  \vec R(s(t)) $.
 Then, 
 $$
  \left\|{d \vec Y \over dt}\right\|^2 = 
  {\ell^2\over L^2} \left( 1 -  \left( {d L \over ds} \right)^2 \right) \left({ds \over dt}\right)^2 + \left({d\ell \over dt}\right)^2,
 $$
 similarly using $\vec R \cdot \vec R = 1$, $ {d \vec R \over ds} \cdot \vec R = 0$, and $\left\|{d \vec X \over ds}\right\|^2 = 1$.
 Therefore $ \left\|{d \vec y \over dt}\right\|^2 =   \left\|{d \vec Y \over dt}\right\|^2 = 1 $ and the mapping is isometric.
\end{proof}

A similar argument works for cylindrical surfaces.
\begin{lemma}\label{cylinder-isometry}
  If an arc-length parameterized crease $\vec x(s)$
  have unique rule segments on one side parallel to $\vec r$,
  such that $\vec r$ is perpendicular to segment $c$.
  then an embedding $\vec f$ is a proper folding if and only if
  folded curve $\vec X = \vec f \circ \vec x$ is also arc-length parameterized, and
  the perpendicular rule segments from $\vec x(s)$ to $c$ maps isometrically to 
  rule segments from $\vec X(s)$ perpendicularly to a planar curve $C$, where $C = \vec f\circ c$.
\end{lemma}
\begin{proof}
 Necessity (only if part) is obvious, so we prove the sufficiency (if part).
 The folded curve is arc-length parameterized by $s$ as  ${d \vec X(s) \over ds} = {d \vec x(s) \over ds} = 1$, 
 and the length of ruling segment $L(s)$ must be equal $L(s) = \|\vec x(s) - \vec c(s)\| = \|\vec X(s) - \vec C(s)\|$.
 Let $\vec r$ denote the unit ruling vectors from the apex toward the curve in 2D, 
 i.e., $\vec x(s) = \vec c(s) + L(s)\vec r$.
 Similarly denote unit ruling vector in 3D by $\vec X(s) = \vec C(s) + L(s)\vec R$.
 Consider a coordinate system using arc length $s$ and length along the ruled segments $\ell$.
 The face is uniquely ruled between the crease and the curve at any point, so $(s, \ell)$ uniquely represent a point on the portion.
 A point $(s, \ell)$ in 2D corresponds to $\vec c(s) + \ell \vec r$,
 which is mapped to 3D to $\vec C(s) + \ell \vec R$.
 Consider a 2D $C^1$ curve $\vec y(t)$ represented by $(s(t), \ell(t))$, where $t$ is the arc-length parameterization.
 Then the total derivative of $\vec y(t) = \vec c(s) + \ell(t)  \vec r $ is
 $$
  {d \vec y \over dt}
  = {\partial \vec y \over \partial s}{ds\over dt} + {\partial \vec y \over \partial \ell} {d \ell\over dt} 
  = {d \vec c \over ds}{ds\over dt} +  \vec r {d\ell\over dt}.
 $$
 Then,
 \begin{eqnarray*}
  \left\|{d \vec y \over dt}\right\|^2 &=&
  \left\|{d \vec c \over ds}\right\|^2 \left({ds \over dt}\right)^2 
  +  2 {d \vec c \over ds} \cdot \vec r \left({ds \over dt}\right) \left({d\ell \over dt}\right) 
  + \|\vec r\|^2 \left({d\ell \over dt}\right)^2\\
  &=& \left\|{d \vec c \over ds}\right\|^2 \left({ds \over dt}\right)^2 + \left({d\ell \over dt}\right)^2,
 \end{eqnarray*}
 where we used $\vec r \cdot \vec r= 1$ and $ {d \vec c(s) \over ds} \cdot \vec r = 0$. 
 Now differentiate $ L(s) \vec r + \vec c(s) = \vec x(s)$ to obtain
 $$
 {d \vec c \over ds} + {d L \over ds} \vec r = {d \vec x \over ds}.
 $$
 By taking the dot product,
 $$
 \left\| {d \vec c \over ds} \right\|^2 + \left( {d L \over ds} \right)^2 = \left\|{d \vec x \over ds}\right\|^2 = 1,
 $$
 again using $ {d \vec c \over ds} \cdot \vec r = 0$ and $\vec r \cdot \vec r = 1$.
 Thus,
 $$
  \left\|{d \vec y \over dt}\right\|^2 = 
  \left( 1 -  \left( {d L \over ds} \right)^2 \right) \left({ds \over dt}\right)^2 + \left({d\ell \over dt}\right)^2.
 $$
 The mapped crease $\vec Y(t)$ in 3D is defined by $\vec Y(t) = \vec C(s) + \ell(t)  \vec R $.
 Then,
 $$
  \left\|{d \vec Y \over dt}\right\|^2 = 
  \left( 1 -  \left( {d L \over ds} \right)^2 \right) \left({ds \over dt}\right)^2 + \left({d\ell \over dt}\right)^2,
 $$
 similarly using $\vec R \cdot \vec R = 1$, $ {d \vec C \over ds} \cdot \vec R = 0$, and $\left\|{d \vec X \over ds}\right\|^2 = 1$.
 Therefore $ \left\|{d \vec y \over dt}\right\|^2 =   \left\|{d \vec Y \over dt}\right\|^2 = 1$ and the mapping is isometric.
\end{proof}

By Lemmas~\ref{ruled-surface-isometry}~and~\ref{cylinder-isometry}, the existence of the folded form is
ensured by constructing the folded crease $f(\gamma)$ such that in the folded state, 
the distance between the $V_{0,1}$ and $f(\gamma(t))$ is always $r(t)$, and the distance from $zx$ plane is always $\ell(t)$.
If we view the projection of the curve, this is equivalent to constructing a curve represented by polar coordinate $(\theta(t), h(t))$ $(\theta \in \mathbb R)$, such that (i)~the curve has arclength $t$ and (ii)~$\theta(t)$ is a monotonic function (in order to avoid self-intersection).
Condition (i) yields a differential equation
$$
1 =  h^2 \left({d\theta(t)\over dt}\right)^2 + h'(t)^2.
$$
Condition (ii) gives us $0 < {d\theta(t)\over dt}$ and $h(t) > 0$,
so the differential equation becomes
$$
{d\theta(t)\over dt} = {1\over h(t)}\sqrt{1 - \left({d h(t) \over dt}\right)^2},
$$
which has solution
$$
\theta(t) = \int_0^t {1\over h(t)}\sqrt{1 - \left({d h(t) \over dt}\right)^2} dt
$$
if and only if $\left({d h(t) \over dt}\right)^2 \leq 1$ for $t \in (0,1)$.
Combined with condition (ii), $\left({d h(t) \over dt}\right)^2 < 1$.
$$
\left({d h(t) \over dt}\right)^2 = {\left[(t-u) - {1 \over 2} (v-v^*) \ell'(t)\right]^2 \over (t-u)^2 + (v-v^*) \left[{1 \over 4} (v+v^*)-\ell(t)\right]}< 1,
$$
which is equivalent to
$$
-{1\over 4}(v-v^*)\left[1+\left({d \ell(t) \over d t}\right)^2\right] + \left[{1 \over 2} v - \left(\ell(t) + {d \ell(t) \over d t}(u-t)\right)\right] > 0.
$$
Because $\ell(t) + {d \ell(t) \over d t}(u-t) $ represents the $y$ coordinate of the intersection between the tangent line to $\gamma_{0,0}^+$ at $t$ and a vertical line passing through $V_{0,1}$, ${v \over 2} - \left(\ell(t) + {d \ell(t) \over d t}(u-t)\right)$ is always positive. Also $1+\left({d \ell(t) \over d t}\right)^2$ is positive, so the condition is given by
$$
v-v^* < {4 \left[{v \over 2} - \left(\ell(t) + {d \ell(t) \over d t}(u-t)\right)\right] \over 1+\left({d \ell(t) \over d t}\right)^2}.
$$
If we define $v^*_{lim}<v$ as $$v-v^*_{lim} = {4 \left[{v \over 2} - \left(\ell(t) + {d \ell(t) \over d t}(u-t)\right)\right] / \left[1+\left({d \ell(t) \over d t}\right)^2\right]},$$ then there exists a continuous solution for $v^*$ in $(v^*_{lim}, v)$.



Now that we have folded an individual gadget, we can tile the gadget
to get a proper folding of the overall crease pattern.
Here we use the fact that the oriented folded module, for a sufficiently
small fold angle, projects to a kite in the $x y$ plane.


Consider inversions of the oriented folded module through the midpoints of
its boundary edges, followed by negating all normals to swap the top and bottom
sides of the paper (Figure~\ref{fig:kite-connect}).
If we consider the $xy$ projection, the operation corresponds to $180^\circ$
rotation around the midpoint of the kite, resulting in a tessellation.
Thus, in particular, there are no collisions between the copies of the
folded module.
Because each connecting edge is mapped onto itself in 3D,
this tessellation has no gap in 3D.
Also, because the boundary is on a ruled segment,
the surface normal vector is constant along each edge.
The surface normal is flipped by the inversion, and then negated back to its
original vector, so the surface normals at corresponding points match.
Thus the shared boundaries remain uncreased in the tessellated folding.
To show that this tessellated folding comes from one sheet of paper,
we can apply the same tiling transformation to the crease-pattern module,
which is also a kite, so it tiles the plane with the same topology
and intrinsic geometry.
Therefore the plane can fold into the infinitely tiled folding.
\end{proof}

\begin{figure}[tb]
  \centering
  \includegraphics[width=\linewidth]{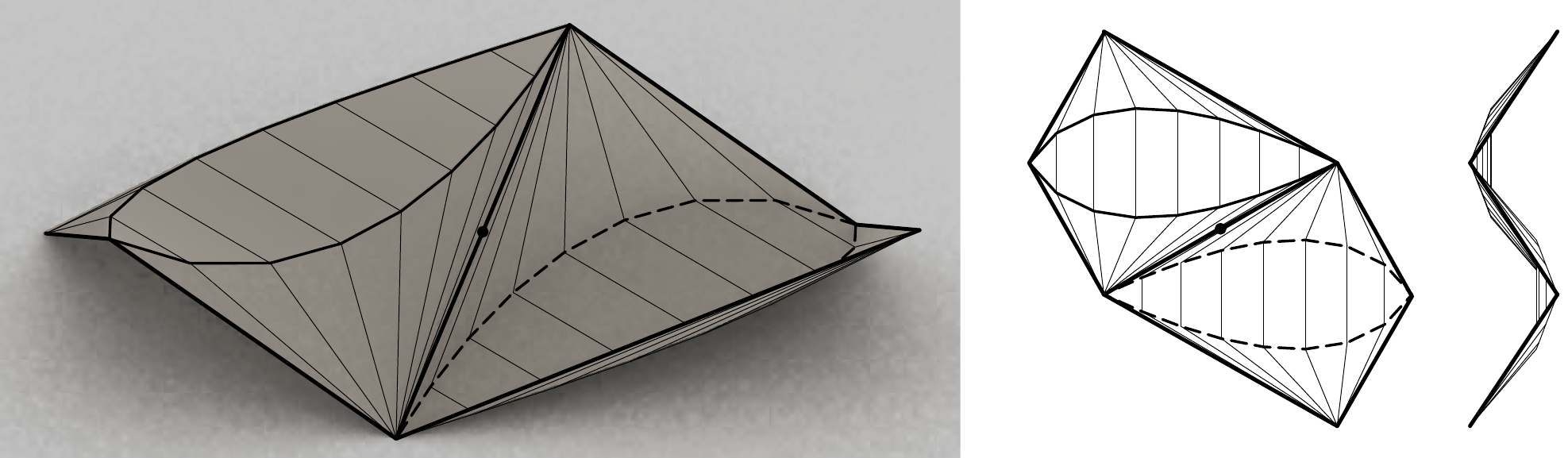}
  \caption{The connection of kite structures.}
  \label{fig:kite-connect}
\end{figure}


\section*{Acknowledgments}

We thank the Huffman family for access to the third author's work,
and permission to continue in his name.



\bibliography{lens}

\begin{thebibliography}{99}
\newcommand{\enquote}[1]{``#1''}

\bibitem[Demaine and O'Rourke~07]{GFALOP}
Erik~D. Demaine and Joseph O'Rourke.
\newblock \emph{Geometric Folding Algorithms: Linkages, Origami, Polyhedra}.
\newblock Cambridge University Press, 2007.

\bibitem[Demaine et~al.~11]{Hypar}
Erik~D. Demaine, Martin~L. Demaine, Vi~Hart, Gregory~N. Price, and Tomohiro
  Tachi.
\newblock \enquote{(Non)existence of Pleated Folds: How Paper Folds Between
  Creases.}
\newblock \emph{Graphs and Combinatorics} 27:3 (2011), 377--397.

\bibitem[Fuchs and Tabachnikov~99]{Fuchs-Tabachnikov-1999}
Dmitry Fuchs and Serge Tabachnikov.
\newblock \enquote{More on Paperfolding.}
\newblock \emph{The American Mathematical Monthly} 106:1 (1999), 27--35.
\newblock Available online (\url{http://www.jstor.org/pss/2589583}).

\bibitem[Fuchs and Tabachnikov~07]{Fuchs-Tabachnikov-2007-developable}
Dmitry Fuchs and Serge Tabachnikov.
\newblock \enquote{Developable Surfaces.}
\newblock In \emph{Mathematical Omnibus: Thirty Lectures on Classic
  Mathematics}, Chapter~4. American Mathematical Society, 2007.

\bibitem[Huffman~76]{Huffman-1976}
David~A. Huffman.
\newblock \enquote{Curvature and Creases: A Primer on Paper.}
\newblock \emph{IEEE Transactions on Computers} C-25:10 (1976), 1010--1019.

\end{thebibliography}
\bibliographystyle{akpbib}

\end{document}